\documentclass{article}

\usepackage{amsmath, amsthm, amssymb, bm, bbm}

\usepackage{graphicx}
\usepackage{verbatim}
\usepackage{natbib}
\usepackage{caption}
\usepackage{subcaption}
\usepackage{fancyvrb}
\usepackage{enumerate}
\usepackage{mdframed}
\usepackage{relsize}
\usepackage{hyperref}
\usepackage[margin=1.5in]{geometry}
\hypersetup{colorlinks,citecolor=blue,urlcolor=blue,linkcolor=blue}

\makeatletter
\newenvironment{algorithm}[1]{
  \protected@edef\@currentlabelname{#1}
  \protected@edef\@currentlabel{#1}
  \begin{mdframed}[
    frametitle={#1},
    frametitlerule=true]
}{
  \end{mdframed}
}
\makeatother

\usepackage{stefan_tex}

\newcommand*{\printdate}{%
   \ifcase \month\or January\or February\or March\or April\or May\or June\or July\or
    August\or September\or October\or November\or December\fi \space \number\year}

\theoremstyle{plain}
\newtheorem{prop}{Proposition}

\newtheorem{theo}[prop]{Theorem}
\newtheorem{assumption}{Assumption} 
\newtheorem{lemma}{Lemma} 
\theoremstyle{definition}
\newtheorem{exam}{Example}
\newtheorem{defi}[exam]{Definition}

\theoremstyle{remark}

\newtheorem{rema}{Remark}

\author{Georgy Kalashnov\thanks{
Contact emails: Georgy Kalashnov (\texttt{gkalashnov@ucsd.edu}),
Evan Munro (\texttt{evan.munro@chicagobooth.edu}), and
Stefan Wager (\texttt{swager@stanford.edu}).} \\
Department of Economics, University of California, San Diego
\and
Evan Munro \\
Booth School of Business, University of Chicago
\and
Hao Sun \\
Didi Chuxing Technology Company
\and
Shuyang Du \\
Meta Platforms, Inc.
\and
Stefan Wager \\
Graduate School of Business, Stanford University}
\date{\printdate}
\title{Treatment Allocation under Uncertain Costs\thanks{
We are grateful for helpful discussions with Dmitry Arkhangelsky,
Susan Athey, Hamid Nazerzadeh, Liyang Sun and  Erik Sverdrup,
and for feedback from seminar participants at Stanford and Uber.}}

\begin{document}

\maketitle

\begin{abstract}
We consider the problem of learning how to optimally allocate treatments whose cost is uncertain and can vary with pre-treatment covariates. This setting may arise in medicine if we need to prioritize access to a scarce resource that different patients would use for different amounts of time, or in marketing if we want to target discounts whose cost to the company depends on how much the discounts are used. Here, we show that the optimal treatment allocation rule under budget constraints is a thresholding rule based on priority scores (those with a higher score are treated first), and we propose a number of practical methods for learning these priority scores using data from a randomized trial. Our formal results leverage a statistical connection between our problem and that of learning heterogeneous treatment effects under endogeneity using an instrumental variable. We find our method to perform well in a number of empirical evaluations.
\end{abstract}

\section{Introduction}

Data-driven resource allocation is increasingly prevalent across a number of fields.
One popular approach starts by modeling treatment heterogeneity. Given
a treatment (or intervention) and an outcome of interest, we also collect a large number of (pre-treatment)
covariates and seek to estimate how these covariates modulate the effect of the treatment on the outcome.
We then allocate treatment to those individuals who are predicted to respond most strongly to it based on their covariates.
As examples of this paradigm, in medicine, \citet*{basu2017detecting} consider assigning more
aggressive treatment to reduce blood pressure to cardiovascular disease patients who are estimated to benefit
from it the most; in marketing, \citet{ascarza2018retention} and \citet{lemmens2020managing} consider targeting retention offers to customers who are
estimated to be most responsive to them; while in economics,
\citet{kitagawa2018should} discuss prioritizing eligibility to job training programs to those job applicants who are
estimated to get the largest employment boost from the program.

One limitation of this line of work, however, is that existing methods for treatment personalization mostly do not consider the cost of assigning treatment. In all three cases considered above, this is not a problem: here,
treating any one specific person costs roughly the same as treating another, and so allocating treatment based on
estimated outcomes alone is valid. However, in many problem settings the cost of treating different people is not the same, and is unknown pre-treatment. When there is a budget constraint limiting the total resources that we can spend on the treatment, determining which individuals to prioritize requires learning the relation of benefits as well as the costs of the treatment to the covariates.

\begin{exam}
\label{exam:marketing}
\textbf{Marketing incentives.}
Suppose a gym wants to evaluate a campaign that gives a ``first month free'' offer to some potential customers, with the goal of enrolling more long-term members.
Clearly, the treatment effect may vary across customers, as may the cost. Some recipients of the offer may visit the gym just a handful of times during their
free month (low cost) and then upgrade to a regular membership (high reward), while others may use the gym's facilities every day during
their free month (high cost) but then fail to convert (low reward). A marketing campaign that allocates resources only based on rewards but not costs may not
spend its budget optimally.
An analysis of a marketing experiment with this structure run by a sharing economy company is available in the working paper version of this paper \citep{sun2021treatmentWP}.
\end{exam}

\begin{exam}
\label{exam:crisis}
\textbf{Targeting scarce healthcare resources in a crisis.}
Consider a hospital that has insufficient intensive care beds to treat all incoming patients,
and needs to choose whom to prioritize given available resources. Suppose, moreover, that the hospital only has
two types of incoming patients. Patients of type A are responsive to treatment, and their chance of survival rises by
10\% if admitted to intensive care; however, their recovery is slow, and they will spend 20 days in the unit if admitted.
In contrast, patients of type B get a 5\% increase in chance of survival if admitted, but will only spend 5 days in the unit
if admitted. Here, targeting based on treatment heterogeneity would prioritize patients of type A, but this is not
the utility-maximizing prioritization rule: If the hospital only targets patients of type A, in the long run it can save
0.5 patients per day per 100 intensive care beds, whereas if it only targeted patients of type B it could double this number to
1 patient per day per 100 intensive care beds.
\end{exam}

\begin{exam}
\label{exam:insurance}
\textbf{Insurance subsidies.}
Suppose a philanthropic organization wants to offer a subsidized insurance product.
The organization has a finite budget, and wants to design its program to maximize
benefits (e.g., in the case of health insurance, to maximize the total improvement
along a target health metric). In this setting, utility-maximization requires
considering both how much a recipient would benefit from the insurance, and
how many claims they might make (and thus how much of the total budget they
would use up).
\end{exam}

In this paper, we propose an approach to optimal treatment prioritization in a setting where we have
a limited budget, and our treatment of interest has costs that are both variable and uncertain. We
show that the optimal feasible treatment rule ranks units by a cost-aware priority score, formed as
a ratio of conditional expected incremental benefits to conditional expected incremental costs, and then treats people ordered by this priority score until budget runs out (or the intervention
is no longer beneficial). We argue that the fact that our approach is priority-based, i.e., that it
first ranks units by priority and then allocates them to treatment until the budget
has been spent, has some notable practical advantages: It ensures monotonicity in
treatment assignment (i.e., the set of people treated at a higher budget level is
a superset of people treated at a lower budget level), and enables us to more precisely
enforce the budget constraint when deploying the policy to new data.

The main learning problem in the paper is estimating the optimal priority scores;
our proposed policies then involve targeting using the estimated score.
We start by showing that in a semi-parametric setting---where the priority score is
linear in the pre-treatment covariates---a moment-based estimator of the score function
converges at a $1/\sqrt{n}$-rate and has an asymptotically normal sampling distribution.
In the more general non-parametric setting, we show that the scores can be learned by using
existing algorithms based on generalized random forests \citep*{breiman2001random,athey2019generalized}.
We also provide a method for inference on the benefit of a given estimated priority-based rule under a specific budget. 

We find our approach to perform well in a number of applications,
and to enable meaningful gains relative to approaches that do not
account for variable costs in targeting.

\subsection{Related Work}

The need to account for the costs of an intervention arises in a number of application areas.
The effectiveness of an intervention across studies is often compared on the basis of cost-effectiveness, i.e., the positive effect for a dollar invested.
\citet{hendren2020unified} perform a meta-study in which they compare a large number of experiments with public expenditures on the basis of cost-effectiveness,
and also discuss a common and sensible way to construct the costs and the benefits variables across studies. \citet{dhaliwal2013comparative} do the same focusing
on education. However, while such cost-effectiveness comparisons across interventions or treatments are ubiquitous in the literature, these papers do not generally
consider the heterogeneity in the cost-effectiveness estimates within their study in a systematic way, or the potential for targeted treatments.

Our contribution fits broadly into a growing literature on treatment personalization, including \citet*{bertsimas2019optimal},
\citet*{hahn2020bayesian}, \citet{kallus2018confounding}, \citet{kennedy2020optimal}, \citet{kunzel2019metalearners},
\citet{nie2020quasi}, \citet{wager2018estimation}, \citet{zhao2012estimating}.
Most of this literature has focused on settings where cost of treatment
is constant across units and so doesn't enter into considerations about
optimal targeting; however, there are a handful of recent exceptions,
involving two general approaches to taking costs
into account for treatment personalization.
Both approaches solve the same optimization problem of maximizing outcomes while constraining costs to meet a budget, but use algorithms that are not priority-based.

The first approach, considered by \citet*{hoch2002something} and \citet{xu2020estimating},
is to create a new outcome, called the net monetary benefit, which captures both the cost
and benefit of treatment. Concretely, this approach specifies outcomes of the
form ``$\text{reward} - \nu \times \text{cost}$'' and then runs standard methods
for personalization of these outcomes. This approach is helpful if we are
able to pre-commit to a value of $\nu$ that brings costs and rewards to the
same scale.  The second approach, considered in \citet{huang2020estimating}, \citet{sun2021empirical}
and \citet*{wang2018learning}, is to directly impose cost constraints into
the outcome-weighted learning approach of \citet{zhao2012estimating}.
This approach is conceptually direct and is amenable to extensions, such as multiple treatments, which are not straightforward using a priority-based rule. However, it relies on a non-trivial
optimization problem that can be difficult to solve with many thousands of observations. For both optimization methods, enforcing a specific target budget exactly is not feasible, and considering interventions over a range of possible budgets requires re-fitting the model each time. In contrast,  our approach relies on ranking by priority scores, and the budget only impacts the cutoff above which individuals are treated. This means that the treatment can be rolled out sequentially until the deployment budget is exhausted, and the performance of the rule can be evaluated at multiple budget levels using a single estimate of the priority scores. Furthermore, our approach is amenable to inference on the `lift' of the policy, e.g. the expected gain of the policy compared to a random policy spending the same budget, which is not available for the existing methods. \citet{wang2018learning} also recognize that the solution to the optimization problem takes the form of Theorem \ref{theo:pi_opt} in this work, and propose what they refer to as a regression-model-based
learning algorithm. This estimator is equivalent to the direct ratio approach discussed and used as a baseline method in the simulations in this paper. Our paper, however, goes beyond the results in \citet{wang2018learning} by proposing priority-based estimators that improve upon the direct ratio estimator, and deriving asymptotic guarantees on the performance of the treatment rule. 
Finally, \citet{luedtke2016optimal}  and \citet{bhattacharya2012inferring} discuss the role of budget when allocating treatments; however, they
assume a constant cost of treatment.

\section{Optimal Allocation under Budget Constraints}

Throughout this paper, we formalize causal effects using the potential outcomes framework \citep{imbens2015causal}.
We assume that we observe independent and identically distributed tuples $(X_i, \, W_i, \, Y_i, \, C_i) \simiid P$ for $i = 1, \, \ldots, \, n$,
where $X_i \in \xx$ denotes pre-treatment covariates, $W_i \in \cb{0, \, 1}$ denotes treatment assignment,
$Y_i \in \RR$ denotes the observed outcome, and $C_i \in \RR$ denotes incurred cost. 
Here, both $Y_i$ and $C_i$ depend on the assigned treatment $W_i$, and we capture this relationship
via potential outcomes: We posit pairs $\cb{Y_i(0), \, Y_i(1)}$ and $\cb{C_i(0), \, C_i(1)}$ denoting the outcomes
(and respectively costs) we would have observed for treatment assignments $W_i = 0$ and $W_i = 1$, such that
we in fact observe $Y_i = Y_i(W_i)$ and $C_i = C_i(W_i)$ given the realized treatment $W_i$.  In many applications,
we may know a priori that $C_i(0) = 0$ (i.e., there is no cost to not assigning treatment); for now, however, we also
allow for the general case where $C_i(0)$ may be non-zero. 
Throughout, we assume that treatment increases costs in the following sense:

\begin{assumption}
\label{as:c}
Treatment is costly, in that $C_i(1) \geq C_i(0)$ almost surely
and $\mathbb E[C_i(1) - C_i(0) \,|\, X_i = x] > 0$ for all $x \in \xx$.
\end{assumption}

\label{sec:opt}

The goal is to use the sample of data $(X_i, W_i, Y_i, C_i)$ for $i = 1, \ldots, n$ to estimate the optimal treatment allocation rule. The first step is to define the optimal treatment allocation rule in the population $P$ under a budget constraint and variable costs. A treatment allocation rule (or policy) is a function $\pi: \xx \rightarrow [0, \, 1]$ mapping pre-treatment covariates to an action, where prescriptions $0 < \pi(x) < 1$ are interpreted as random actions
(i.e., we randomly assign treatment with probability $\pi(x)$).
The (incremental) value\footnote{In most cases, expected outcomes are an appropriate measure of value for a policymaker. As acknowledged elsewhere in the causal inference literature, however, in cases with heavy-tailed data, or where the policy-maker is concerned about avoiding harm to certain subgroups, then alternative objectives should be considered \citep{manski1988ordinal}. }   $V$ of a policy $\pi$ is the expected gain it achieves
by treating the units it prescribes treatment to, $V(\pi) = \EE{\pi(X_i) \p{Y_i(1) - Y_i(0)}}$, while the (incremental)
cost $G$ of $\pi$ is $G(\pi) = \EE{\pi(X_i) \p{C_i(1) - C_i(0)}}$.
Given a budget constraint $B$,\footnote{We choose to impose the budget constraint on the incremental cost of the policy over the baseline cost of treating nobody, $\mathbb E[C_i(0)]$. In cases where a policymaker prefers a budget on total cost, then they may prefer to replace $B$ with $\bar B = B + \mathbb E[C_i(0)]$.} the optimal policy $\pi^*_B$ solves the following knapsack-type problem
\begin{equation}
\label{eq:pi_opt_orig}
\pi^*_B := \arg\max \left\{V(\pi) : G(\pi) \leq B \right\}.
\end{equation}
Recall that the knapsack problem involves selecting a set of items such as to maximize the aggregate ``value''
of the selected items subject to a constraint on the allowable ``weight''; and, in our setting, the treatment
effect $Y_i(1) - Y_i(0)$ is the value we want to maximize while the incremental cost $C_i(1) - C_i(0)$ acts as a weight.
There is a key difference between our treatment allocation problem and the traditional knapsack problem. We do not know the distribution of the outcomes or costs, and need to learn them from data. Here, we
momentarily abstract away from the learning problem and first write down the form of the optimal treatment
assignment rule given the true data generating distribution; then, we will turn towards learning in the following sections.

In this setting, the form of the optimal treatment allocation rule \eqref{eq:pi_opt_orig} follows directly from the
well known solution to the fractional knapsack problem \citep{dantzig1957discrete}. The optimal policy involves
first computing the following conditional cost-benefit ratio
function,\footnote{In the medical literature, this quantity is also known
as the incremental cost-effectiveness ratio \citep*{hoch2002something}. We use
the convention that $a/0$ is equal to $+\infty$ if $a > 0$, $-\infty$ if $a < 0$, and $0$ if $a = 0$.}
\begin{equation}
\label{eq:ratio}
\rho(x) := \frac{\EE{Y_i(1) - Y_i(0) \cond X_i = x}}{\EE{C_i(1) - C_i(0) \cond X_i = x}},
\end{equation}
and then prioritizing treatment in decreasing order of $\rho(x)$.
The following result formalizes this statement.
The proof of Theorem \ref{theo:pi_opt} given in the appendix generalizes
an argument from \citet*{luedtke2016optimal} to the setting with variable
costs.  \citet{wang2018learning} prove a version of Theorem \ref{theo:pi_opt} when  $\rho(X_i)$ has a continuous distribution, so a deterministic policy is optimal, and acknowledge the possibility of a randomized rule in the more general setting handled explicitly here. 
\begin{defi}
  A policy $\pi_B$ is a (stochastic) threshold policy based on the score $s: \mathcal{X} \to \mathbb{R}$ if there exists a threshold $\rho_B$ and  $a_B \in [0, 1]$ such that
  \begin{equation} 
\label{eq:pi_opt}
\pi_B(x) =   \{ 
0 \text{ if }  \ s(x) < \rho_B,
\ a_B  \text{ if }  \ s(x) = \rho_B, 
\ 1 \text{ if } \ s(x) > \rho_B.
\}
\end{equation}
\end{defi}

\begin{theo}
\label{theo:pi_opt}
Under Assumption \ref{as:c}, the optimal (stochastic) policy $\pi^*_B$  is a threshold policy based on the score $\rho(x)$ from \eqref{eq:ratio} with threshold $\rho_B$ and randomization parameter $a_B$. Additionally, either $\rho_B = a_B = 0$ (i.e., we have sufficient budget to treat everyone with a positive treatment effect), or $\rho_B > 0$ and this policy has cost exactly $B$ in expectation. In the case where $\rho(X_i)$ has a bounded density, $\mathbb{P}[\rho(X_i) = \rho_B] = 0$, the policy $\pi^*_B$ is both deterministic and the unique optimal policy.
\end{theo}

\begin{rema}
We emphasize that $\pi^*_B$ involves ranking units by the ratio of conditional expectations $\rho(x)$,
rather than by the actual cost-benefit ratios $R_i = (Y_i(1) - Y_i(0)) \,\big/\, (C_i(1) - C_i(0))$, as one might
expect in a classical deterministic knapsack specification where the value and cost of each unit is known. In our setting, any candidate policy gives the same treatment probability to all units with a given value of observables $X_i$. This means that the value and cost of treating units with a given value of $X_i$ is respectively $\mathbb E[Y_i(1) - Y_i(0) | X_i ]$ and $\mathbb E[C_i(1) - C_i(0) | X_i]$. This means that the bang-for-buck measure that solves the knapsack problem is the ratio of conditional expectations, rather than $R_i$, which is not identifiable, or $\mathbb E[R_i | X_i]$, which is not the correct aggregation over unit-level heterogeneity. \end{rema}

\subsection{Identifying the Priority Score in Randomized Trials}
\label{sec:identification}

To make use of Theorem \ref{theo:pi_opt} in practice, we
need to make assumptions that let us identify the target $\rho(x)$ from observable
data. The difficulty here is that $\rho(x)$ depends on all four potential outcomes $Y_i(0)$, $Y_i(1)$, $C_i(0)$
and $C_i(1)$, whereas we only get to observe the realized outcomes $Y_i = Y_i(W_i)$ and $C_i = C_i(W_i)$.
Such difficulties are recurrent in the literature on treatment effect estimation, and arise from what
\citet{holland1986statistics} calls the fundamental problem of causal inference.

Here, we address this difficulty by assuming that we have access to data from a randomized controlled trial, i.e.,
where $W_i$ is determined by an exogenous random process; or, more generally, that we have data where the treatment
assignment mechanism is unconfounded in the sense of \citet{rosenbaum1983central}, i.e., that it is
as good as random once we condition on pre-treatment covariates $X_i$.
Randomized controlled trials are frequently used to guide treatment allocation decisions in application
areas where costs may matter \citep[see, e.g.,][]{banerjee2011poor,gupta2020maximizing,kohavi2009controlled},
and unconfoundedness assumptions are widely used in the literature on treatment personalization
\citep{kunzel2019metalearners,wager2018estimation}.

The following result shows how, under unconfoundedness, we can re-write $\rho(x)$ in terms of observable moments.
Given this result, the problem of estimating $\rho(x)$ now reduces to a pure statistical problem of
estimating a ratio of conditional covariances.

\begin{assumption}
\label{assu:unconf}
The treatment assignment mechanism is unconfounded, \\
$\sqb{\cb{Y_i(0), \, Y_i(1), \, C_i(0), \, C_i(1)} \indep W_i} \cond X_i,$
and satisfies overlap, $0 < \PP{W_i = 1 \cond X_i = x} < 1$.
\end{assumption}

\begin{prop}
\label{prop:unconf}
In the setting of Theorem \ref{theo:pi_opt}, if Assumption \ref{assu:unconf} holds then $
\rho(x) = \frac{\operatorname{Cov}[Y_i, \, W_i \mid X_i = x]}{\operatorname{Cov}[C_i, \, W_i \mid X_i = x]}$.
\end{prop}

At first glance, the problem of estimating a ratio of covariances as in Proposition \ref{prop:unconf} may seem like an
explicit but potentially difficult statistical problem. However, there is a useful connection between the
statistical task of estimating this ratio, and the task of estimating a (conditional) local
average treatment effect using an instrumental variable \citep*{angrist1996identification,durbin1954errors}.
Specifically, suppose we have independent and identically distributed samples $(X_i, \, Y_i, \, T_i, \, Z_i)$
where the $X_i$ are pre-treatment covariates, $T_i$ is a (potentially endogenous) treatment, $Y_i$ is an outcome, and $Z_i$ is an
(exogenous) instrument. In this setting and under further assumptions discussed in \citet*{imbens1994identification},
the (conditional) local average treatment effect, $\lambda(x) = \Cov{Y_i, \, Z_i \cond X_i = x}/\Cov{T_i, \, Z_i \cond X_i = x}$
is a natural measure of the causal effect of the endogenous treatment $T_i$ on the outcome $Y_i$.
Several authors, including \citet{abadie2003semiparametric}, \citet{angrist2008mostly},
\citet{chernozhukov2018double}, \citet*{athey2019generalized} and \citet*{wang2018instrumental},
have then used this instrumental variables setting as motivation for developing methods that boil down
to estimating a ratio of conditional covariances  as in $\lambda(x)$.

The upshot is that, although our problem and that of treatment effect estimation with instruments
are conceptually very different, they both reduce to statistically equivalent ratio estimation problems:
Despite divergent derivations and motivations, there is no difference between the statistical targets
$\rho(x)$ in Proposition \ref{prop:unconf} and $\lambda(x)$. Thus, we can take any method for estimating $\lambda(x)$, and turn it into an estimator for $\rho(x)$ by simply
plugging in our treatment $W_i$ where the method expects an ``instrument'' $Z_i$, and plugging in
our cost $C_i$ where it expects a ``treatment'' $T_i$.

\section{Learning Treatment Allocation Rules}
\label{sec:learning}

The simple characterization of the optimal treatment rule $\pi^*_B$ given in Theorem \ref{theo:pi_opt} suggests
the following simple algorithm for treatment prioritization, in cases where the score has bounded density, so the optimal policy is unique and deterministic: 1) get an estimate $\hrho(x)$ of the ratio \eqref{eq:ratio} on a training set where pre-treatment covariates $X_i$, treatment $W_i$, and realized costs and outcomes $(Y_i, C_i)$ are observed, 2) on new data (e.g. a deployment set), rank units $i$ in descending order of $\hrho(X_i)$, and treat those with estimated ratio above the estimated threshold $\hat \rho_B$: $\hat \pi(X_i) = \mathbbm{1}( \hrho(X_i) > \hat \rho_B)$. In other words, each individual is assigned a priority score, and the estimate of this priority score will not depend on the budget. Individuals are assigned to the treatment in order of their priority, up until a threshold, where the threshold ensures the budget constraint is respected. 

Depending on the setting, we may apply different thresholds. Below we consider two examples of thresholds achieving different objectives. In each, we will apply the same threshold formula to different data samples: 
\begin{equation}
  \label{eq:threshold}
  \hat \rho_B= \min \left \{ p \in  [0, \infty) :  \frac{1}{n_\text{sample}} \sum_{i=1}^{n_\text{sample}} \gamma_i C_i \mathbbm{1} \{  \hrho(X_i) \geq p \} \leq B \right \},
\end{equation}
where $\gamma_i$ are weights that, for example, allow us to estimate the budget spent for threshold $p$ using the holdout dataset: $\gamma_i = \frac{W_i}{\mathbb{E}[W_i \mid X_i]} - \frac{(1- W_i)}{1 - \mathbb{E}[W_i \mid X_i]}$.

In some settings we can apply this formula to compute the threshold that (nearly) exactly satisfies the budget in a deployment set. Consider one important setting when this is possible: the control arm has no cost (i.e., $C_i(0) = 0$), we
have an upper bound on the treatment costs, $C_i(1) \leq M$, and the treatment
cost is immediately (or rapidly) revealed for units if they're assigned to
treatment. The priority score creates treatment assignment rules, which are monotonic in the budget. For $B' \geq B$, and any deployment set of individuals, any individual that is treated under budget $B$ is also treated under $B'$.  This allows a budget for a campaign to be increased after the campaign has already started. In this case, we can satisfy the budget to within tolerance $M/n_\text{deployment}$
by: treating units in descending order of $\hrho(X_i)$ and keeping track of
the accumulated costs from treated individuals; and then stopping
when the accumulated cost of treatment is within $M/n_\text{deployment}$ of $B$. This procedure is analogous to using a threshold given by \eqref{eq:threshold} applied to the deployment set, where $\gamma_i$ is set to 1.

In settings where there is a significant delay in observing realized costs after treatment this algorithm is not feasible and controlling the realized costs on the deployment set will in general not be possible, so it is necessary to use a separate holdout set to compute the threshold. Then, the budget is met in expectation asymptotically, but there may be finite-sample violations of the constraint. The approach in \citet{sun2021empirical} can be applied to adjust the threshold so that the budget constraint is satisfied with high probability in finite samples. 

The following result guarantees that the priority-based approach will translate accurate estimates of $\hrho(\cdot)$ into low-regret treatment-assignment rules.
This result accommodates both settings where our learned policy $\hat{\pi}_B(\cdot)$ exactly respects its budget constraint ($\delta = 0$) or over/under-spends its budget
($\delta > 0$ or $\delta < 0$ respectively).

\begin{theo}
\label{theo:rate}
Assume that the optimal policy is unique and deterministic. Under the setting of Theorem \ref{theo:pi_opt}, suppose that $\hrho(x)$ is a score function. Then, if we use a threshold policy $\hat{\pi}_B(x) = \mathbbm{1}[\hrho(x) > \hrho_B]$, that spends approximately the same budget as the optimal policy with a possible budget violation $\delta = \mathbb{E}_T[(\pi^*_B(X_i) - \hat{\pi}_B(X_i))\tau_C(X_i)]$, the resulting regret will be bounded as
$V(\pi^*_B) - V(\hat{\pi}_B) \leq \mathbb{E}_T[\tau_C(X_i)|\rho(X_i) - \hrho(X_i)|] + \hrho_B \delta$, where $\tau_C(X_i) = \mathbb{E}[C_i(1) - C_i(0) \mid X_i]$ and $\mathbb{E}_T[\cdot]$ is an expectation over an independent draw of the data, holding the training data that is used to estimate $\hrho(\cdot)$ fixed. 
\end{theo}

\begin{rema}
Even if we exactly satisfy the budget in the deployment set, the term $\hrho_B \delta$ is defined in terms of the population distribution and will be non-zero. When it is optimal to spend the whole budget and $\rho_B > 0$, then both the threshold based on the holdout set and the threshold based on the deployment set will give $\hrho_B \delta = O_p(\delta) = O_p(1/\sqrt{n_\text{sample}})$, where $\text{sample} \in \{ \text{holdout}, \text{deployment} \}. $  
\end{rema}

Now, to make use of this framework, it remains to develop estimators for $\rho(x)$.
First, in Section \ref{sec:parametric}, we consider a semi-parametric
specification where $\rho(x)$ is assumed to be linear in $x$,
but the conditional covariances $\Cov{Y_i, \, W_i \cond X_i = x}$ and
$\Cov{C_i, \, W_i \cond X_i = x}$ themselves may have a complex dependence on $x$. While linearity is a restrictive assumption, it leads to a simple algorithm with good performance in simulations, even when scores are non-linear, and there are some practical examples where an unknown scaling function enters both conditional outcomes and costs, but $\rho(x)$ is linear. The results in Section \ref{sec:parametric} can also be extended to more complex  parametric approximations to $\rho(x)$ at a cost of a more complex estimator and proof. 

In the linear setting, we develop a Neyman-orthogonal estimator for $\rho(x)$ that
allows for $1/\sqrt{n}$ rates of convergence. In practice, implementing the estimator requires to follow \ref{algo}, which is an implementation of two-stage least squares regression with data-splitting. 
Second, in Section \ref{sec:nonparametric}, we propose a
non-parametric estimator for $\rho(x)$ based on random forests. The last paragraph of this section explains how to implement this estimator in practice using the R package \texttt{grf}.

\subsection{Parametric Estimation of the Priority Score}
\label{sec:parametric}

To understand the fundamental nature of the problem of estimation of $\rho(x)$, we
start by considering a semiparametric model following \citet{robinson1988root} where
$\rho(x)$ is constrained to be linear, $\rho(x) = x^\prime\beta$, but the rest of the
model is left non-parametric. This representation leads to a method-of-moments type estimator that has the
same form as a just-identified instrumental variables estimator, and gives a transparent lens on
the key drivers of asymptotic accuracy. Although the linearity assumption is a strong assumption in many practical settings, understanding the performance of the estimator in the parametric setting is helpful before turning to the non-parametric setting.

Here, we follow the ``double machine learning'' approach to semiparametric estimation following \citet{chernozhukov2018double}.
We start by defining a score function: $e_i(\beta, h(X_i)) = (W_i - h_w(X_i)) \left[(Y_i - h_y(X_i)) - (C_i - h_c(X_i)) \rho(X_i)\right]$, with $h_w(x) := \EE{W_i \mid X_i = x}$, $h_y(x) := \EE{Y_i \mid X_i = x}$, and $h_c(x) := \EE{C_i \mid X_i = x}$, and
note that the identification result in Proposition \ref{prop:unconf}
is equivalent to the score function being mean-zero at the true value of $\beta$ (the details are in Appendix \ref{appendix}),
\begin{equation}
\label{eq:conditional_mom}
\EE{e_i(\beta, h(x) )| X_i = x} = 0 \text{ for all } x \in \xx.
\end{equation}
The terms $h(x)$ are nuisance components, i.e., unknown functions that are not of direct interest, but
are required to form the score functions. However, the score function can be verified to be Neyman orthogonal, i.e., the moment condition \eqref{eq:conditional_mom} is robust to small errors in the nuisance components: For any perturbation function $\delta(x)$,
$\sqb{\frac{d}{ d \varepsilon} \EE{e_i(\beta, \, h(x) + \varepsilon \delta(x)) | X_i = x}}_{\varepsilon = 0} = 0$,  for all $x \in \xx$.
Details are given in the proof of Theorem \ref{thm:dml}.
As argued in \citet{chernozhukov2018double}, this Neyman-orthogonality property is crucial to estimators motivated by \eqref{eq:conditional_mom}
enabling robust inference about $\beta$ using flexibly estimated nuisance functions.

Now, the conditional moment \eqref{eq:conditional_mom} is restricted at each value $x \in \xx$,
and may be difficult to work with in practice if the $X_i$ have continuous support or are high dimensional. However, this condition also implies that, given
$\mathcal{B} = \cb{\beta' : \EE{X_i e_i(\beta', h(X_i) )} = 0}$, we must have $\beta \in \mathcal{B}$, and that if $\mathcal{B}$ is a singleton then it identifies $\beta$. We make use of this fact\footnote{This construction is not the only way to turn \eqref{eq:conditional_mom}
into a practical, unconditional moment restriction. In fact, \citet{chernozhukov2018double} shows that, writing
$\sigma^2(x) = \mathbb E[e_i(\beta, h(X_i))^2 | X_i = x]$ and $R(x) = \mathbb E \left[\frac{\partial}{\partial \beta} e_i(\beta, h(X_i) )| X_i = x\right]$,
then the moment condition $\EE{\sigma^{-2}(X_i) R(X_i) e_i(\beta, h(x))} = 0$ leads to a semi-parametrically efficient estimator of $\beta$, reaching the \citet{chamberlain1992efficiency} efficiency bound. However,
estimating $\sigma^2(x)$ and $R(x)$ leads to additional complexity, and so we rely on the simple linear form here.},
along with the cross-fitting of the nuisance components \citep{schick1986asymptotically} to estimate $\beta$ as in \ref{algo}. 
We show below that this estimator achieves a parametric rate of convergence for $\beta$ provided the nuisance
components $\hat{h}$ converge reasonably fast (but not necessarily at a parametric rate themselves), and the
moment condition is full rank. Our proof follows from general results developed in \citet{chernozhukov2018double}.

\begin{figure*}[t]
\begin{algorithm}{Algorithm 1}
    \label{algo}
    \begin{enumerate}
    \item Randomly split the training data into $K$ equally sized folds $k: \mathbb{N} \to \{1, \, \ldots, \, K\}$.
    \item For each fold $k = 1, \, \ldots, \, K$, produce an estimate of the nuisance components $\hat{h}^{(-k)}(\cdot)$ using
    data in all but the $k$-th folds.\footnote{To do this, one can use the \texttt{regression\_forest} command of the \texttt{grf} package in R, for example. In some cases the researcher might have access to ground truth propensity scores and can use them instead of estimating $\hat h_w(\cdot)$.}
    \item Run a two-stage least squares algorithm, instrumenting a regression of $Y_i - \hat h^{(-k(i))}_y(X_i)$ on $(C_i - \hat h^{(-k(i))}_c(X_i)) X_i$ with $(W_i - \hat h^{(-k(i))}_w(X_i)) X_i$ to output $\hat \beta$
    \end{enumerate}
\end{algorithm}
\end{figure*}

\begin{assumption}
\label{assu:nuisance_rates}
Assume that $\xx \subseteq \mathbb{R}^m$. We use estimators $\hat{h}$ of $h$ for which the following hold. There exists a sequence $\delta_n \rightarrow 0$ and constants $a > 0 , A > 0$ and $q > 4$ such that,
when trained on $n$ IID samples from our generative distribution $P$, we obtain an estimator $\hat{h}$ satisfying,
with probability tending to 1 as $n$ gets large,

\begin{equation}
\begin{split}
& \EE[T]{(\hat h_y(X_i) - h_y(X_i))^2}^{\frac{1}{2}} \leq  \rho_{y, n}, \ \ \ \
\EE[T]{(\hat h_y(X_i) - h_y(X_i))^q}^{\frac{1}{q}} \leq A, \\
& \EE[T]{(\hat h_c(X_i) - h_c(X_i))^2}^{\frac{1}{2}} \leq  \rho_{c,n} , \ \ \ \
\EE[T]{(\hat h_c(X_i) - h_c(X_i))^q}^{\frac{1}{q}} \leq A \\
& \EE[T]{(\hat{h}_w(X_i) - h_w(X_i))^2}^{\frac{1}{2}} \leq \rho_{w, n}, \ \ \ \
\EE[T]{(\hat h_w(X_i) - h_w(X_i))^q}^{\frac{1}{q}} \leq A 
\end{split}
\end{equation}
with $\rho_{w, n} \rho_{c, n} \leq \frac{\delta_n}{n^{1/2}}$, $\rho_{w, n} \rho_{y, n} \leq \frac{\delta_n}{n^{1/2}}$ and also $\rho_{w, n} < \delta_n, \rho_{c, n} < \delta_n, \rho_{y, n} < \delta_n$. 
\end{assumption}

\begin{assumption}
\label{assu:regularity} Outcomes, costs, and covariates are bounded, so  there is a finite constant $A > 0$ such that Assumption \ref{assu:nuisance_rates} holds and $X_i \in [-A, A ]$, $| Y_i | \leq A$ and  $| C_i| \leq A$. Additionally, let ${ V}_i = W_i -  \mathbb E[W_i | X_i]$, ${ D}_i  = ( C_i  - \mathbb E[C_i | X_i])$
and $U_i =  Y_i - \mathbb E[Y_i | X_i] -  D_iX'_i \beta$. Assume there is a constant $a > 0$ such that $\EE{V_i^2U_i^2 |X_i} \geq a$.
\end{assumption}

\begin{theo} \label{thm:dml}
Under the assumptions of Proposition \ref{prop:unconf}, suppose furthermore that Assumption \ref{assu:nuisance_rates} and Assumption \ref{assu:regularity}
hold. Additionally, $\mathbb E[X_i X_i']$ is full rank. Then, our estimator $\hbeta$ described above satisfies $\sqrt n \p{\hat { \beta} - \beta} \Rightarrow_d N\p{0, \bm V_{\beta}}$, $\bm V_{\beta} = \mathbb E[ V_i  D_i X_iX_i']^{-1} \mathbb E[U^2_i V_i^2 X_i X_i'] \mathbb E[V_i  D_iX_iX_i']^{-1}$.
\end{theo}

The key property of Theorem \ref{thm:dml} is that we get $1/\sqrt{n}$-rate convergence for $\hbeta$ even if the
rest of the problem is not parametrically specified. In particular, the numerator and denominator used to define
$\rho(x)$ in Proposition \ref{prop:unconf}, i.e., $\Cov{Y_i, \, W_i \cond X_i = x}$ and $\Cov{C_i, \, W_i \cond X_i = x}$,
need not admit a linear specification. Rather, it's enough to be able to estimate relevant nuisance components
at slower rates, e.g. $\hat {h}(X_i) - h(X_i) = o_p(n^{-1/4})$, and this can be done via flexible machine
learning methods.

As discussed in Theorem \ref{theo:rate}, the asymptotic distribution of the estimator $\hbeta$ can be used to bound the regret of the threshold policy based on the estimated score. Making use of Theorem \ref{thm:dml} we have 
$\sqrt{n}(\hat \rho(x) - \rho(x)) \Rightarrow_d N(0, x^{\prime} \bm V_\beta x)$, so $\sqrt{\mathbb{E}_T[(\hat \rho(X_i) - \rho(X_i))^2]} = O_p(n^{-1/2})$. Since costs are bounded, this (by Jensen's inequality) implies that, assuming we are in a setting, where we can satisfy the budget exactly, the regret will also converge at a $n^{-1/2}$ rate.

\subsection{Non-Parametric Estimation of the Priority Score}
\label{sec:nonparametric}

If we're willing to assume that $\rho(x)$ admits a linear form, then the estimator discussed above achieves
good large-sample performance. However, in many applications, we may not be willing to assume a linear
specification $\rho(x) = x'\beta$, and instead seek a non-parametric estimator for $\rho(x)$. In this case,
one possible approach would be to first separately estimate the numerator and denominator in Proposition \ref{prop:unconf},
$\Cov{Y_i, \, W_i \cond X_i = x}$ and $\Cov{C_i, \, W_i \cond X_i = x}$, and then form
$\hrho(x) = \hCov{Y_i, \, W_i \cond X_i = x} \, \big/ \, \hCov{C_i, \, W_i \cond X_i = x}$.
This approach, however, is potentially suboptimal: If the numerator and denominator are more complex than
$\rho(x)$, then the rates of convergence we could achieve via this approach would be slower than ones
we could get via directly targeting $\rho(x)$ \citep{foster2019orthogonal,nie2020quasi}.

Here, we consider one particular solution to direct estimation of $\rho(x)$ based on the ``generalized
random forest'' framework of \citet*{athey2019generalized}. Generalized random forests provide an approach
to turn any conditional moment restriction for a target parameter, such as \eqref{eq:conditional_mom}, into an estimator
for the target parameter that adapts the popular random forest method of \citet{breiman2001random}.
The key idea of the algorithm is that it grows a forest specifically designed to express heterogeneity
in $\rho(x)$, and can thus be more responsive to the actual complexity of this function than methods
that estimate $\Cov{Y_i, \, W_i \cond X_i = x}$ and $\Cov{C_i, \, W_i \cond X_i = x}$ separately
and then take the ratio of these two estimates.

Like random forests, the approach starts by growing a set of $B$ decision trees by recursive partitioning
on the covariates $X_i$ on subsamples of size $s = n^\beta$, where $0<\underline{\beta} < \beta < 1$\footnote{$0 < \underline{\beta}$ depends on the details of the partitioning algorithm, which are discussed by \cite{athey2019generalized}.}. For each tree indexed $b = 1, \, \ldots, \ B$ and a given test point $x$, let $L_b(x)$ be a set of data points falling within the same leaf as $x$ in a tree $b$. Let us define
forest weights $\alpha_i(x) = \frac{1}{B} \sum_{b = 1}^B \frac{1\p{\cb{i \in L_b(x)}}}{\sum_{j = 1}^n 1\p{\cb{j \in L_b(x)}}}$.
Conceptually, the weights $\alpha_i(x)$ capture the relevance of each observation $i = 1, \, \ldots, \, n$
for estimation at $x$; formally, we note that the usual regression forest prediction at $x$ can
be expressed as a weighted average of outcomes $Y_i$ with weights $\alpha_i(x)$.
In our setting, generalized random forests estimate\footnote{In some settings, the budget may require distinguishing between subgroups with heterogeneous treatment effects on outcomes that each have a small treatment effect on costs. In these cases, the estimate of $\hat \rho(x)$ may be unstable due to the small denominator, and a practitioner may benefit from regularizing the denominator of the estimator, for example by adding a small, positive constant to the observed cost for all treated observations.} $\rho(x)$ by solving an empirical version of
\eqref{eq:conditional_mom} with the forest weights $\alpha_i(x)$:
\begin{equation}
\label{eq:ivf}
\begin{split}
&\hrho(x) = \frac{\sum_{i = 1}^n \alpha_i(x) \p{Y_i - \bar h_y(X_i)}\p{W_i - \bar h_w(X_i)}}{\sum_{i = 1}^n \alpha_i(x) \p{C_i - \bar h_c(X_i)}\p{W_i - \bar h_w(X_i)}},
\end{split}
\end{equation}
where  $\bar h_y(x) = \sum_{i = 1}^n \alpha_i(x) Y_i$, $\bar h_c(x) = \sum_{i = 1}^n \alpha_i(x) C_i$, $\bar h_w(x) = \sum_{i = 1}^n \alpha_i(x) W_i$.

As discussed in \citet*{athey2019generalized}, it is helpful to compare \eqref{eq:ivf}
to a simpler $k$-nearest neighbors estimator that first discards all but the $k$ closest observations
to $x$ in covariate space, and then estimates $\rho(x)$ by solving an unconditional version of 
\eqref{eq:conditional_mom} on those $k$ observations. From the perspective of this comparison,
the advantage of generalized random forests is that the weights $\alpha_i(x)$ provide a well tuned,
data-adaptive notion of neighbors relevant to estimating $\rho(x)$.

We refer to \citet*{athey2019generalized} and \citet{athey2019estimating} for details,
including a discussion of how the recursive partitioning used to grow the individual trees
in the forest is run. At a high level, the trees are grown to greedily express as much heterogeneity
as possible in $\rho(x)$. These works also detail how subsampling and subsample splitting are used
to stabilize the estimator, which is important for the formal result given in \citet*{athey2019generalized}. This formal result applies directly to our setting, and ensures large-sample consistency of the learned $\hrho(x)$ under the conditions of Proposition \ref{prop:unconf} and an additional assumption of Lipschitz continuity, which replaces the linearity assumption:

\begin{assumption}[Lipschitz continuity]
  \label{as:lipschitz}
  $\EE{W_i \mid X_i = x}$, $\EE{Y_i \mid X_i = x}$, $\EE{C_i \mid X_i = x}$, $\EE{Y_i W_i\mid X_i = x}$, $\EE{C_i W_i\mid X_i = x}$ are Lipschitz continuous in $x$.
\end{assumption}

A verification of the assumptions in \citet*{athey2019generalized} and the application of their Theorem 5 gives the following result on the asymptotic distribution of $\hrho(x)$:

\begin{theo}
  \label{thm:grf}
   Under the assumptions of Proposition \ref{prop:unconf}, suppose furthermore that Assumptions \ref{assu:regularity} and \ref{as:lipschitz} hold. Then, there is a sequence $\sigma_n(x)$, for which $\frac{\hrho(x) - \rho(x)}{\sigma_n(x)} \Rightarrow N(0, 1)$, and $\sigma_n^2(x) = \operatorname{polylog}(n/s)^{-1} \frac{s}{n}$,
where $\operatorname{polylog}(n/s)$ is a function that is bounded away from 0 and increases at most polynomially with the log-inverse sampling ratio $\log(n/s)$.
\end{theo}

Finally, from a practical perspective, we can again make use of the formal connection to instrumental variables
estimation here. Although the specification above would be enough to build a generalized random forest
for estimating $\rho(x)$, doing so would seem to require a non-trivial amount of implementation work.
However, it turns out that the calculations required to estimate $\rho(x)$ are exactly the same as are
already performed in the ``instrumental forest'' method provided in the \texttt{grf} package of
\citet*{athey2019generalized}, and so we can re-purpose this function for our use case. Specifically,
 we use instrumental forests to estimate $\rho(x)$ by replacing the method's inputs $Z_i$ and
$T_i$ with $W_i$ and $C_i$ respectively (we pass covariates $X_i$ and the outcome $Y_i$ to the instrumental
forest as usual).

\section{Evaluating the Performance of a Targeting Rule}
\label{sec:costcurve}

In deciding whether or not to implement a given estimated treatment prioritization rule, it is useful to characterize for a fixed budget how much the population is expected to benefit in expectation from prioritization compared to a uniform rule, which assigns treatment randomly. In this section, we show how to evaluate the performance of a fixed targeting rule, $\hat \rho(\cdot)$, which is estimated on a training dataset, when the performance is measured on a holdout set. As in \cite{yadlowsky2021evaluating}, this does not account for randomness in score estimation when deriving the distribution of the policy value; this choice corresponds to the uncertainty faced by a policymaker that estimates a priority rule once, and uses that rule on new data. 

  For this section, to keep notation and exposition simple, we assume that the score has a continuous distribution with bounded density, so that there is a unique and deterministic solution to Theorem \ref{theo:pi_opt}. Under this assumption, we can define the expected per-person (incremental) value and budget of a given treatment rule directly in terms of a threshold $s$. $V_{\hat \rho}(s) = \mathbb E_T[ (Y_i(1) - Y_i(0)) \mathbbm{1}(S_i \geq s) ]$ and $G_{\hat \rho}(s) = \mathbb E_T[(C_i(1) - C_i(0)) \mathbbm{1}(S_i \geq s) ]$, where $S_i = \hat \rho(X_i)$, and $\mathbb E_{T}[\cdot ]$ indicates that the expectation is taken over an independent draw of the data, and is conditional on the estimated priority score $\hat \rho(\cdot)$. 

Let $b \cdot n$ define a budget constraint for a sample of $n$ individuals. The expected gain from spending a budget of $b \cdot n$ under a uniform rule is $ b \cdot \tau_y/  \tau_c$, where $b/ \tau_c$ is the fixed treatment probability induced by a budget of $b$ under a uniform rule, and $\tau_y = \mathbb E[Y_i(1) - Y_i(0)]$ and $\tau_c = \mathbb E[C_i(1) - C_i(0) ]$ are ATEs.  The lift at a given budget $b$ is the difference in the expected gain in outcomes from a fixed priority rule compared to a uniform randomized rule that spends the same budget. Computing the lift is useful for evaluating whether there is sufficient heterogeneity in treatment response conditional on $X_i$; it assists a policymaker in deciding whether to implement a more complex priority-based treatment rule instead of a simple lottery allocation policy. We can formally define the lift of a priority rule for a given budget $b$ as $\Delta_{\hat \rho}(b) = Q_{\hat \rho}(b) -  b \cdot \tau_y/  \tau_c$,
where $Q_{\hat \rho} (b) = V_{\hat \rho}( G_{\hat \rho}^{-1}(b))$ is the expected value of the fixed priority rule at budget $b$, and the inverse of $G_{\hat \rho}(s)$ exists by the assumptions in Theorem \ref{thm:Qexp}.
 
We next describe how to perform estimation and inference on the value of the prioritization rule and its lift over a uniform rule at a single budget value, and how to construct a QINI curve using these estimated values. The QINI curve is a popular visualization that, for a family of thresholded scoring rules, plots the cost of treatment on
the $x$-axis and the benefit of treatment on the $y$-axis \citep{ascarza2018retention,imai2019experimental,rzepakowski2012decision,yadlowsky2021evaluating}.

Existing results on estimating QINI curves, however, assume that the cost of treating each unit is the same, and so the cost of treatment on
the $x$-axis is equivalent to the number of units treated; however, in our setting, this equivalence no longer holds.

To address this challenge, we propose the following estimator for the QINI curve in a setting with uncertain costs. We first form
inverse-propensity weighted estimators of $V_{\hat \rho}(s)$ and $G_{\hat \rho}(s)$ with a holdout sample of size $n_{holdout}$  as follows:  $\hat V_{\hat \rho}(s) = \frac{1}{n_{holdout}} \sum \limits_{i=1}^{n_{holdout}} \left (\frac{W_i}{h_w(X_i)} - \frac{(1- W_i)}{1 - h_w(X_i)} \right ) Y_i \mathbbm{1} \{  \hrho(X_i) \geq s \}, \hat G_{\hat \rho}(s) = \frac{1}{n_{holdout}} \sum \limits_{i=1}^{n_{holdout}} \left (\frac{W_i}{h_w(X_i)} - \frac{(1- W_i)}{1 - h_w(X_i)} \right ) C_i \mathbbm{1} \{  \hrho(X_i) \geq s \}$,
where $h_w(X_i) = \PP{W_i = 1 \cond X_i = x}$ is the treatment probability for units with $X_i = x$ (in a uniformly randomized trial, $h_w(x) = q$ would be constant);
these are unbiased for $V_{\hat \rho}(s)$ and $G_{\hat \rho}(s)$ by the randomization of $W_i$ \citep{imbens2015causal}. In this section, we assume that the propensity score is known, but it is possible to follow techniques in the existing literature to extend Theorem \ref{thm:Qexp} to settings where the propensity score is estimated; see, for example,  \citet{hirano2003efficient}, \citet{wooldridge2007inverse} and \citet{graham2012inverse}. 

We plot the curve $(\hat V_{\hat \rho}(S_{i_k}), \, \hat G_{\hat \rho}(S_{i_k}))$ for $k = 1, \, \ldots, \, n_{holdout}$, where $S_{i_1} \geq \ldots \geq S_{i_{n_{holdout}}}$ are the ordered
scores on the holdout set, where $S_i = \hat \rho(X_i)$. Figures \ref{fig:simu} and \ref{fig:application2} illustrate this approach in applications.
The point at which this curve intersects the vertical line at $x = b$ corresponds to an estimate of the lift that can be achieved with budget $b$.

We can also construct estimators for $\Delta_{\hat \rho}(b)$ and  $Q_{\hat \rho}( b)$. Let $\hat s( b) \in \hat G^{-1}(b)$ and $\hat \tau_y$ and $\hat \tau_c$ be any consistent estimators of the average treatment effects on outcomes and costs. Then, $\hat Q_{\hat \rho}(b) = \hat V_{\hat \rho}(\hat s(b))$, and $\hat \Delta_{\hat \rho}( b) =  \hat Q_{\hat \rho}(b) -  b \frac{\hat \tau_y}{\hat \tau_c}$. 

In order to derive an inference strategy, our first result is that we can write $\hat Q_{\hat \rho}(b) - Q_{\hat \rho}(b)$ in an asymptotically linear form.

\begin{theo} \label{thm:Qexp}
Under Assumption \ref{as:c} and \ref{assu:unconf}, if we have a scoring rule $S : \xx \rightarrow \mathcal S$ such that $\mathcal S$ is compact,
$V_{\hat \rho}(s)$ and $G_{\hat \rho}(s)$ are continuously differentiable in $s$, the score distribution has strictly positive and bounded density\footnote{We make this assumption for convenience and to keep notation simple rather than being essential. However, we can extend to the setting where the score has mass points by introducing an auxiliary score variable that combines the original score and a uniformly distributed random variable at the mass point, and working with the auxiliary variable instead.} $f(s)$ for all $s \in \mathcal S$, and
there is an approximate inverse in finite samples, i.e., $ \hat G_{\hat \rho}(\hat s(b))  - b = o_p(n^{-0.5})$,
then $\hat Q_{\hat \rho}(b) $ and $\hat \Delta_{\hat \rho}(b)$  have asymptotically linear representations:
\begin{align*}
& \sqrt n \p{\hat Q_{\hat \rho}(b)  - Q_{\hat \rho} (b)} = \frac{1}{\sqrt n} \sum \limits_{i=1}^n \psi^q_i + o_p(1) \Rightarrow_d N(0, \sigma^2_q),  \\
& \sqrt n \p{ \hat \Delta_{\hat \rho} (b) - \Delta_{\hat \rho}(b)} = \frac{1}{\sqrt n} \sum \limits_{i=1}^n \left(\psi^q_i  - \frac{b}{\tau_c} \psi^y_i + b \frac{\tau_y}{\tau_c^2} \psi^c_i\right) + o_p(1) \Rightarrow_d  N(0, \sigma^2_d).
\end{align*}
where $\psi^q_i = V_i(s(b)) - V_{\hat \rho}(s(b)) - \frac{V'_{\hat \rho}(s(b))}{G'_{\hat \rho}(s(b))} (G_i(s(b)) - G_{\hat \rho}(s(b)))$, $\psi^y_i =  \left( \frac{W_i}{h_w(X_i)} - \frac{1- W_i}{ 1- h_w(X_i)} \right) Y_i - \tau_y$, $\psi^c_i = \left( \frac{W_i}{h_w(X_i)} - \frac{1- W_i}{ 1- h_w(X_i)} \right) C_i - \tau_c$. We also have $V_i(s) =  \left( \frac{W_i}{h_w(X_i)} - \frac{1- W_i}{ 1- h_w(X_i)} \right) Y_i \mathbbm {1}(S_i \geq s)$ and 
$G_i(s) =  \left( \frac{W_i}{h_w(X_i)} - \frac{1- W_i}{ 1- h_w(X_i)} \right) C_i \mathbbm {1}(S_i \geq s)$.
\end{theo}

The asymptotic linear representation in Theorem \ref{thm:Qexp} implies asymptotic normality of the estimators. It also implies that various resampling-based
estimators \citep{efron1982jackknife} yield
valid confidence intervals for $\Delta_{\hat \rho} (b)$ \citep{chung2013exact,yadlowsky2021evaluating}.
In particular, Lemma 12 of \citet{yadlowsky2021evaluating}
implies that the half-sample bootstrap will yield valid inference in this setting. We use this
result to justify confidence intervals in our applications below.
We emphasize that these confidence statements are conditional on the training set, i.e., we take the prioritization
 rules learned on the training set as given, and only quantify holdout set uncertainty in estimating the QINI curve. The continuous differentiability of $V_{\hat \rho}(s)$ and $G_{\hat \rho}(s)$ required for this result is satisfied in settings where the score $S_i$ has a continuously differentiable distribution function and both $\mathbb E[Y_i(1) - Y_i(0) | S_i = s]$ and $\mathbb E[C_i(1) - C_i(0) | S_i =s]$ are continuous functions in $s$.

Given a method for estimating the QINI curve in the setting with uncertain costs, we can also estimate the area under the QINI curve, known as the QINI coefficient. The QINI coefficient provides a single metric by which we can judge the performance of an allocation rule in a budget-independent way. The QINI coefficient is the area between the estimated reward of the treatment allocation rule and the random treatment rule with the same cost, as the average budget ranges from $0$ to the average cost of treating everyone in the sample,
$ QINI = \int_{0}^{\tau_c} \Delta_{\hat \rho}(b) db $.

The natural plug-in estimator for this quantity is
\smash{$ \widehat{QINI} = \int_{0}^{\hat \tau_c} \hat \Delta_{\hat \rho}(b)  db$},
where \smash{$\hat \Delta_{\hat \rho}(b)$} is as given above. We believe it plausible
that the result from Theorem \ref{thm:Qexp} can also be extended to provide
a central limit theorem for the QINI coefficient (see also the discussions in
\citet{yadlowsky2021evaluating}); however, we leave this question to further work.

\section{Simulation Study}
\label{sec:simulation}

In order to understand numerical aspects of treatment allocation with uncertain costs,
we conduct a simulation-based comparison of 6 methods for targeting. We consider 4
priority-based methods and two direct optimization methods, proposed
by \citet{hoch2002something} and \citet{sun2021empirical}, that are not priority based. In all our experiments, there is no cost to withholding treatment (i.e., $C_i(0) = 0$) and
we have data from a randomized trial with $\PP{W_i = 1} = p$.
All methods below will make use of these facts whenever appropriate.

\textbf{Ignore Cost.}
We ignore cost, and simply score observations using an estimate $S_i = \htau_Y(x)$ of the treatment effect
$\tau_Y(x) = \EE{Y_i(1) - Y_i(0) \cond X_i = x}$. We estimate $\htau_Y(x)$ using causal forests as implemented
in the \texttt{R}-package \texttt{grf} \citep*{athey2019generalized}.

\textbf{Direct Ratio.}
Our second baseline builds on the characterization result from Theorem \ref{theo:pi_opt}, rather than on the
connection to instrumental variables estimation from Proposition \ref{prop:unconf}. We start by
estimating $\tau_Y(x)$ using causal forests as above, and we also estimate the conditional cost function
$\tau_C(x) = \EE{C_i(1) - C_i(0) \cond X_i = x} = \EE{C_i \cond X_i = x, \, W_i = 1}$ by using a regression
forest from \texttt{grf} to predict $C_i$ from $X_i$ for treated units. Finally, we score
observations using $\hrho(X_i) = \htau_Y(X_i) \,/\,\htau_C(X_i)$.

\textbf{Linear Parametric.} 
For the parametric version of our proposed method, we use \ref{algo} to fit $\hat \beta$, and we use $\hat \beta$ to produce estimated priority scores $\hat \rho(X_i) = X_i' \hat \beta$.

\textbf{Instrumental Forest.} 
For the nonparametric version of our proposed method, as described in Section \ref{sec:nonparametric}, our proposed method gets estimated priority scores $\hrho(X_i)$ from an instrumental forest with ``remapped'' inputs. We use the instrumental forest from the package  \texttt{grf}, except where the function expects an
``instrument'' we provide $W_i$, and where the function expects a ``treatment'' we provide $C_i$.

\textbf{Hoch et al. [2002]}.
 The method predicts a linear combination of the reward and the cost $m(x) = \EE{Y_i(1) - Y_i(0) - \lambda (C_i(1) - C_i(0)) | X_i = x}$ for an appropriate choice of the coefficient $\lambda$ to satisfy the budget constraint. An individual is treated whenever $\hat m(x) > 0$. We use a causal forest from the \texttt{grf} package to estimate $m$. In practice, to meet a specific budget constraint, the $\lambda$ parameter should be chosen by splitting the training dataset, which can add additional noise. For the simulation in Section \ref{sec:simulation}, $\lambda$ is chosen in advance to meet the budget constraint in expectation using a separate sample of data from the data-generating process of the same size as the deployment data.
 
 \textbf{Sun [2021]}. In Table \ref{tab:simu}, we also include a version of \cite{hoch2002something} that applies the approach in \cite{sun2021empirical} for settings where the researcher would like to satisfy the budget constraint in the deployment set with high probability. We first estimate the standard deviation of the budget spent by \cite{hoch2002something} in a deployment set in simulation. Then we choose a larger $\lambda$ to ensure that the budget  in expectation is less than the target constraint minus 1.96 times the standard deviation of the deployment budget.

We first compare the above methods using a simple simulation study that highlights
the behavior of the methods under consideration. For this experiment, we generate covariates
and potential outcomes as follows with $k = 12$ (where left unspecified, variables are generated independently): $X_{ij} \sim \text{Unif}\p{-1, \, 1}$ for $j = 1, \, ..., \, k$, $W_i \sim \text{Bern}\p{p}$, $\varepsilon_i \sim \nn\p{0, \, 1}$,
$Y_i(w) = \max\cb{X_{i1} + X_{i3}, \, 0} + \max\cb{X_{i5} + X_{i6}, \, 0} + w e^{X_{i1} + X_{i2} + X_{i3} + X_{i4}} + \varepsilon_i$,
where $\text{Unif}(a, \, b)$ is a uniform distribution on the interval $[a, \, b]$,
$\nn\p{\mu, \, \sigma^2}$ is a Gaussian distribution with mean $\mu$ and variance $\sigma^2$, and
$\text{Bern}(p)$ stands for the Bernoulli distribution with success probability $p$.
We also consider two settings for the cost $C_i(1)$ of treating a unit: One baseline setting where cost is
random but unpredictable ($C_i(1) \cond X_i \sim \text{Pois}(1)$), and another where cost can be anticipated in terms of covariates ($C_i(1) \cond X_i \sim \text{Pois}\p{e^{X_{i2} + X_{i3} + X_{i4} + X_{i5}}}$), 
where $\text{Pois}(\mu)$ is a Poisson distribution with mean $\mu$. We run both simulations on training sets
of size $n \in \{200, 500, 1000\}$ and with treatment randomization probability $p = 0.5$. We report the $n = 500$ results in the main document and the other variants in the Appendix \ref{appndx:empirical}. The results are qualitatively similar and change the magnitudes of performance gaps between the methods.

\begin{figure}
\centering
\begin{tabular}{cc}
\includegraphics[width=0.45\textwidth]{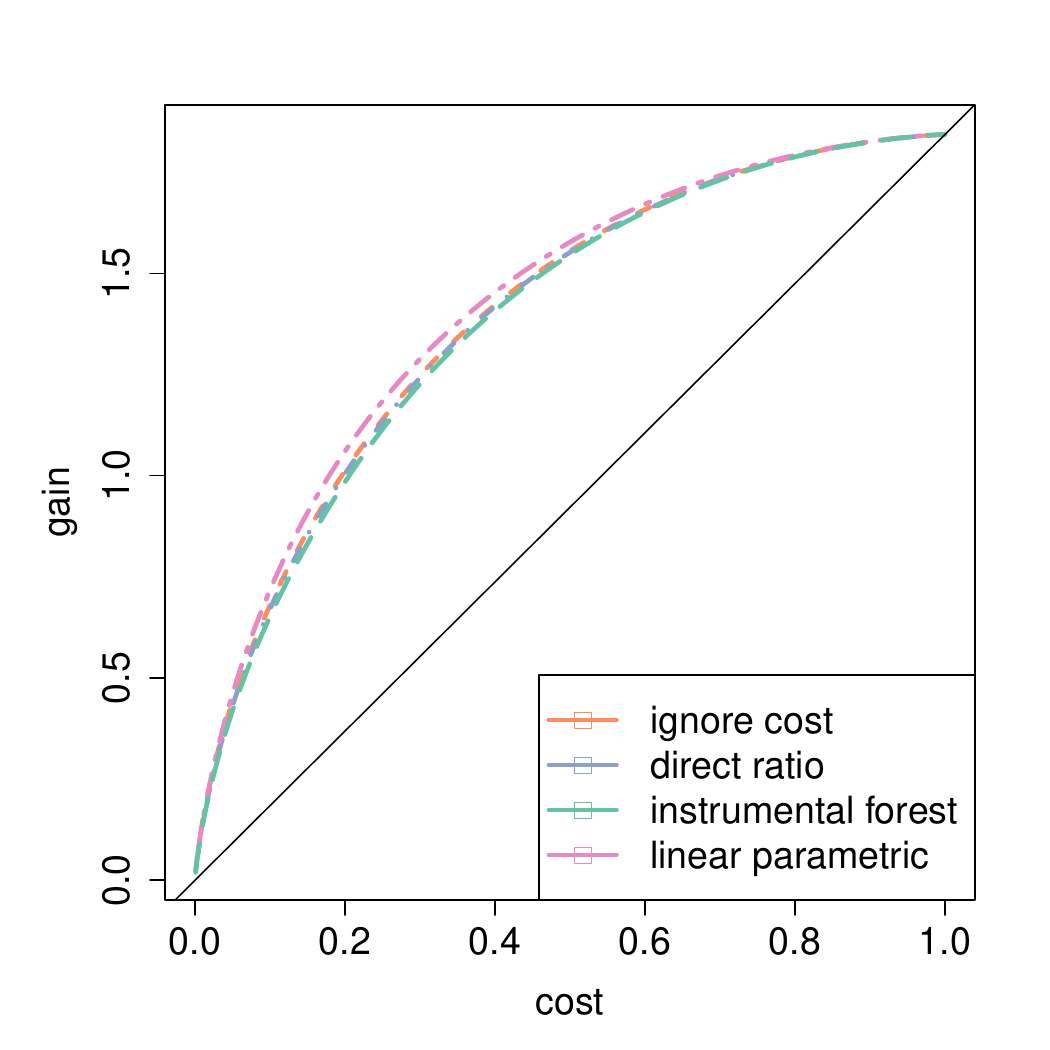} &
\includegraphics[width=0.45\textwidth]{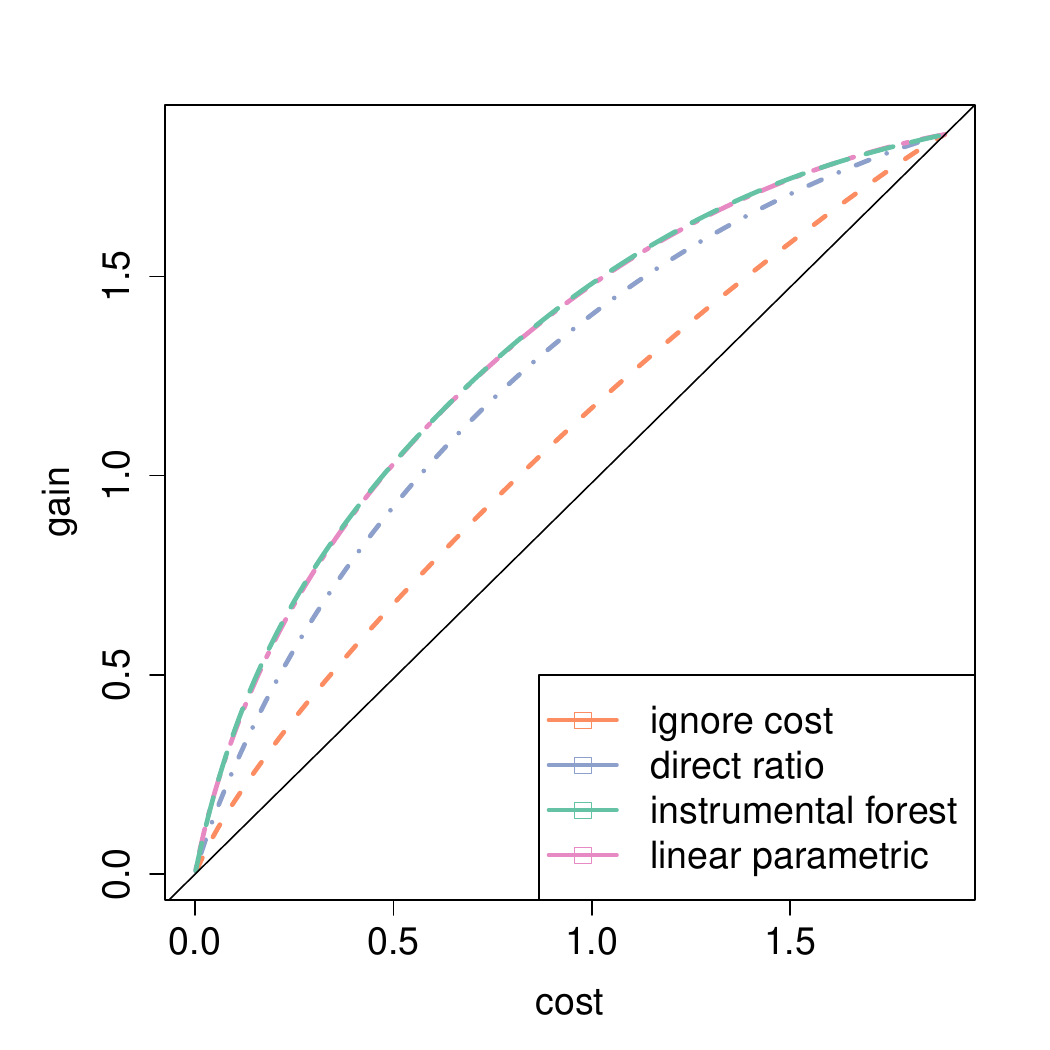} \\
unpredictable cost & partially predictable cost
\end{tabular}
\caption{QINI curves for the simulation settings with predictable and unpredictable costs, averaged over
500 simulation replicates. For each replicate, each method is trained on $n = 500$ samples. The curves are computed for a deployment set of  $n_{deployment} = 10,000$ samples, shared across simulation replicates. For each deployment point $i = 1, \, \ldots, \, n_{deployment}$ we compute the ground truth $\tau_Y(X_i) = \EE{Y_i(1) - Y_i(0) \cond X_i}$,
and the expected cost $\tau_C(X_i) = \EE{C_i(1) \cond X_i}$. Then, given any treatment rule derived from
the training set, we rank the deployment set in decreasing order of the scores used by the treatment rule, and
compute an estimate $\bar R(S_{i_k}) = \frac{1}{n_{deployment}} \sum \limits_{i=1}^{n_{deployment}} \tau_Y(X_i) \mathbbm{1}(S_i \geq S_{i_k}) $ and $\bar B(S_{i_k}) = \frac{1}{n_{deployment}} \sum \limits_{i=1}^{n_{deployment}} \tau_C(X_i)\mathbbm{1}(S_i \geq S_{i_k})$ as cumulative sums along that ranking from $i_{1}, \ldots, i_{n_{deployment}}$.
The above displays are obtained by computing one such QINI curve for each simulation replicate, interpolating
these QINI curves, and then (vertically) averaging the interpolated curves.}
\label{fig:simu}
\end{figure}

In order to evaluate the quality of these treatment rules, we consider results in terms of the QINI curve $Q_{\hat \rho}(b)$ described in Section \ref{sec:costcurve} that maps different
possible budget levels to the value we can get using the considered policy at this budget level.
Figure \ref{fig:simu} compares average deployment set performance of the different priority-based methods in terms of their QINI curves.
In the left panel, with unpredictable costs, there is no visible difference between the four methods. This is as expected,
as the optimal strategy is simply to prioritize units in decreasing order of
$\tau_Y(x) = \EE{Y_i(1) - Y_i(0) \cond X_i = x} = e^{x_1 + x_2 + x_3 + x_4}$. Another setting where ignoring costs can perform reasonably well is when there is negative correlation between $\tau_Y(X_i)$ and $\tau_C(X_i)$. 
In the second setting above, however, there is a divergence between the treatment effect $\tau_Y(x)$ (which remains the same),
and the cost-benefit ratio $\rho(x) = e^{x_1 - x_5}$ we should use for prioritization, and this is reflected in the performance
of different methods. Here, the ``ignore cost'' baseline is targeting the wrong objective, and so performs poorly. The ``direct ratio'' baseline is targeting the correct objective and does better, but still does not match the performance of our proposed method which is designed to focus on $\rho(x)$. Finally, though the priority score is not linear in the data generating process, the parametric method still performs well. 

We note that, here, the function $\tau_Y(x)$ and $\tau_C(x)$ are somewhat aligned, and the induced cost-benefit ratio
function $\rho(x) = \tau_Y(x) / \tau_C(x)$ takes a simpler form than either $\tau_Y(x)$ or $\tau_C(x)$ on its own; specifically
units with large values of $x_2$ or $x_3$ have large values of both $\tau_Y(x)$ and $\tau_C(x)$, and these effects cancel
each other out. This type of structure may arise when there is some group of units that are overall just very responsive to
treatment, in a sense where they both produce considerable value but also incur large costs; and instrumental forests are
well positioned to take advantage of such structure as they can purely focus on fitting $\rho(x)$. In other settings, where
$\tau_Y(x)$ and $\tau_C(x)$ vary in more unrelated ways, the ``direct ratio'' baseline may also be a reasonable candidate for
learning $\rho(x)$.

Computing QINI curves for the methods from the related literature which are not priority-based is computationally difficult, since it requires resolving an optimization problem for each possible budget value on the curve. To compare the performance of the priority-based methods to those in the related literature, we describe results at a fixed budget constraint of 1 in Table \ref{tab:simu}, which is in Appendix \ref{appndx:empirical}. The instrumental forest has the highest lift at this budget level, while the method of \citet{hoch2002something} performs similarly well in terms of lift. The direct ratio approach performs slightly worse. Since the method of \citet{hoch2002something} only meets the budget on the deployment set in expectation, it often violates the budget. The approach of \cite{sun2021empirical} remedies this problem by ensuring the budget is met with high probability on the deployment set, rather than in expectation, but comes at the cost of performance. In contrast, the priority-based methods always spend the correct budget. Furthermore, confidence intervals for the lift computed using the bootstrap have coverage close to 0.95 for the priority-based methods, as expected from the results in Section \ref{sec:costcurve}.

\section{Oregon Health Insurance Experiment}
\label{sec:application2}

In 2008, Oregon conducted a lottery for a limited number of spots in its Medicaid program \citep{finkelstein2012oregon}.
The authors enriched the data on lottery signups with surveys and administrative data and found positive effects of health insurance on self-reported health outcomes, health care utilization,
and financial well-being. This dataset allows us to analyze how a government could optimize a self-reported health outcome under a constraint on Medicaid expenses, for example which depend on the utilization of health services.

For the purpose of our method, the target ``reward'' variable $Y_i$ is self-reported health, which we encode as a binary variable, where $1$ maps to `good', `very good' or `excellent' and 0 maps to `bad' or `fair'.
Meanwhile, we consider two possible ``cost'' variables $C_i$: the number of outpatient visits in the treatment group $C_i$, and the number
of prescribed drugs in the treatment group. We consider the costs $C_i$ to be zero in the control group and non-negative in the treatment group, since we consider our constraint to be on the resources used in the Medicaid expansion.

The baseline survey includes all of the lottery winners as well as an approximately equal number of lottery losers, which amounts to an
initial sample of 58,405 lottery subscribers. 23,777 subjects completed the endline survey in 12 months after the baseline, allowing us to measure the outcome variables.
A few hundred observations are also lost because of incomplete answers in the endline survey, leaving us with a sample of 18,062 when prescribed medications are the cost variable and 23,119 when outpatient visits are the cost variable. \citet{finkelstein2012oregon} check the balance of covariates in their paper and argue that the attrition is balanced across treatment groups
and doesn't invalidate the experiment. We split the sample equally into a training set and a holdout set, stratifying
the split on the number of household members and the assigned treatment.

\begin{figure}
\centering
\begin{tabular}{cc}
\includegraphics[width=0.45\textwidth]{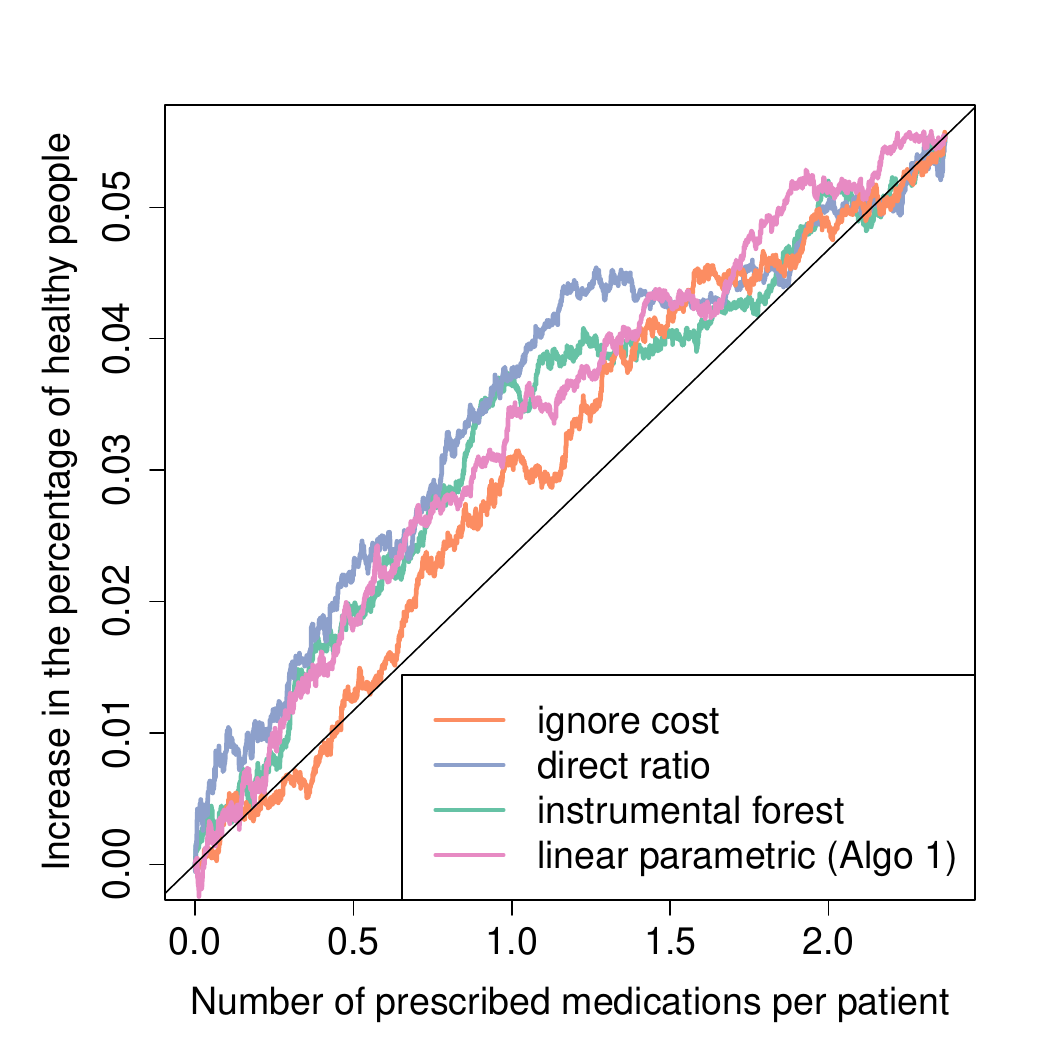} &
\includegraphics[width=0.45\textwidth]{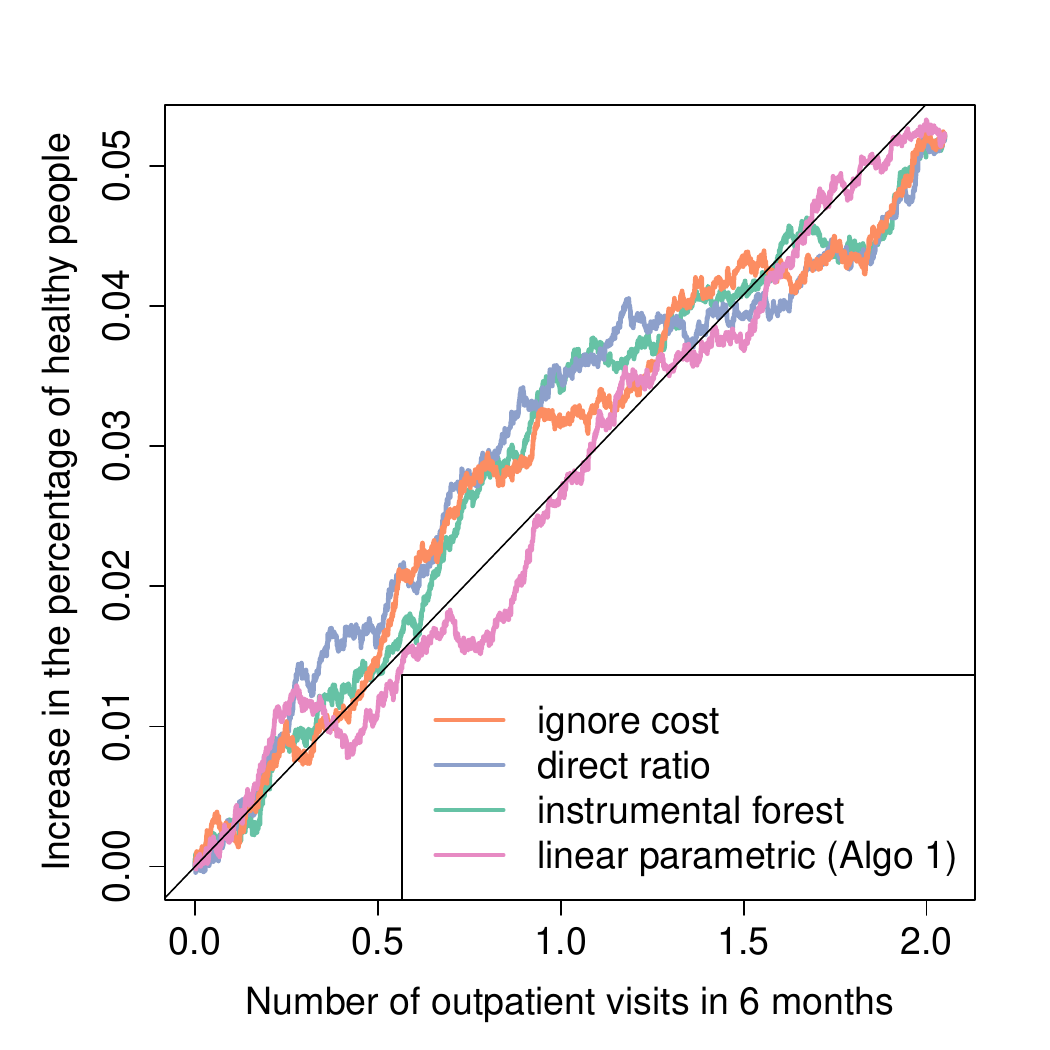}
\end{tabular}
\caption{QINI curves for the Oregon Health Insurance Experiment described in Section \ref{sec:application2}. The total sample size for the left figure is 18062 and 23119 for the right figure, split equally into the holdout and train samples.}
\label{fig:application2}
\end{figure}

Medicaid applies to all family members, while the lottery registrations are individual, therefore the chances of
winning are confounded with the household size $X^\text{(conf)}_i \in \mathbb{Z}$, i.e., members of larger households have a better chance of getting treated; so, we also must estimate the propensity score $h_w(x^\text{(conf)}) = \PP{W_i = 1 \cond X^\text{(conf)}_i = x^\text{(conf)}}$. We use the short demographic characteristics from the registration form, emergency department visits
history and the baseline survey data on demographics, employment, health conditions and past doctor visits to build the model $\hat\rho(X_i)$ of health
improvement per resource usage. We drop some variables from the baseline survey, which could be affected (or are shown in the paper to be affected) by the treatment.
The purpose of this example is to demonstrate the method, therefore we are using all of the available pre-treatment information in learning $\hat\rho$ with instrumental forest.\footnote{When
deploying a method of this type in practice, one would need to audit the covariates used for equity, social and ethical concerns, as well as gameability;
see \citet{athey2020policy} and \citet{kitagawa2018should} for further discussion.} However, we manually select the most powerful predictors for the linear parametric method to avoid curse of dimensionality. The full list of variables in both of the models is included in the Appendix \ref{appndx:empirical}.

We build the QINI curve $\hat Q_{\hat \rho} (b)$ in the same way we did in the previous examples; however, to improve robustness due to using estimated propensity scores $\hat h_w(X^{(conf)}_i)$,
we use a doubly robust adaptation of $\hat V_{\hat \rho}(s)$ and $\hat G_{\hat \rho}(s)$ following \citet{yadlowsky2021evaluating}.
Results are shown in Figure \ref{fig:application2}. In this application, both the instrumental forest and the direct ratio baseline have comparable performance and both noticeably outperform the baseline ``ignore cost'' in the case where we use the number of prescribed medications as a cost variable.

\begin{table} \centering 
\resizebox{\textwidth}{!}{
\begin{tabular}{@{\extracolsep{5pt}}lccccc} 
\\[-1.8ex]\hline 
\hline \\[-1.8ex] 
& IV forest & Direct ratio & Ignore cost & Linear parametric & $\hat Q_{\hat \rho}(1)$ uniform rule \\ 
\hline \\[-1.8ex] 
Medications & 0.0130*** & 0.0143*** & 0.0072 & 0.0113** & 0.0234*** \\
& (0.0050) & (0.0052) & (0.0049) & (0.0052) & (0.0042) \\
\hline \\[-1.8ex] 
Outpatient visits & 0.0086** & 0.0095** & 0.0063 & 0.0013 & 0.0254*** \\
& (0.0042) & (0.0044) & (0.0043) & (0.0045) & (0.0045) \\
\hline
\end{tabular}}

\caption{Lift $\hat \Delta(1)$ relative to random choice, for different prioritization rules and cost variables and bootstrapped standard deviations for them. We also include the $\hat Q_{\hat \rho}(1)$, i.e. the total reward under a budget constraint of 1 under the random choice rule, for reference. The number of observations for the medications outcome variable is 18,062 (9,051 in the holdout sample) and for outpatients visits it is 23,119 (11,602 in the holdout sample) The standard deviations are in parentheses and are clustered at the household level. The stars denote confidence levels: $^{*}$p$<$0.1; $^{**}$p$<$0.05; $^{***}$p$<$0.01}
\label{table:ohie}
\end{table} 
Finally, we also present the estimated lifts $\hat \Delta(1)$ for a chosen budget of 1 prescribed medication or 1 outpatient visit per person. We estimate standard errors using
a bootstrap clustered at the household level. Results are presented in Table \ref{table:ohie}. The instrumental forest and the direct ratio methods significantly
outperform a random choice rule. Conversely, the baseline that ignores costs doesn't give a statistically significant lift for this budget level. Quantitatively, if we have
a budget that allows us to prescribe on average 1 medication per patient among new Medicaid enrollees, then targeting using instrumental forests lets us improve the share of healthy individuals from
2.3\% to 3.6\%. To summarize Figure \ref{fig:application2} using a single metric, we also report the QINI coefficient, as defined in Section \ref{sec:costcurve}, in Table \ref{table:roc} of Appendix \ref{appndx:empirical}. The metric shows that both instrumental forest and the direct ratio method vastly outperform a treatment allocation policy that ignores costs.

\section{Discussion}

In this paper, we considered the problem of optimally targeting a treatment
under budget constraints, while allowing the cost of treating different people to be both
variable and uncertain. Problems with this structure appear frequently in medicine, marketing,
and other areas. Here, we derived the form of the optimal prioritization rule using the solution of \citet{dantzig1957discrete}
to the fractional knapsack problem, and established a statistical connection to the problem of
heterogeneous treatment effect estimation with instrumental variables that allowed us to develop
a number of estimators for the optimal prioritization rule, including one that
re-purposes off-the-shelf random forest software from \citet*{athey2019generalized}.
In the simulation and empirical applications, the proposed approach shows considerable promise in helping us effectively learn whom to treat.

To conclude, we briefly discuss a few possible extensions. Unlike many other works in the policy learning literature, we do not impose functional restrictions on the policy class, and instead impose some restrictions on the complexity of the data-generating process, such as linearity or smoothness of $\rho(x)$.  These restrictions on the data-generating process avoid the negative result in \citet{stoye2009minimax}, even in the absence of constraints on the policy class. Under general functional restrictions, the optimal rule may not have a priority-based structure, and a policymaker could instead solve an empirical version of the constrained optimization problem directly over the policy class. 

With multiple treatments, it is still possible to estimate incremental conditional benefit-cost ratios for each treatment and each individual in the sample. However, since there are multiple ratios for each individual, a priority-based approach no longer follows directly from the estimation of the ratios. Further work is needed to construct a priority-based approach that solves the multiple treatment problem with uncertain costs and benefits.

Last, we provide inference for the QINI curve that is pointwise. A policymaker that wants to use the QINI curve to make a decision, for example on the choice of budget to maximize lift, would instead prefer uniform inference. Providing uniform inference on the QINI curve is also an interesting subject for future work.

\newpage 

\bibliographystyle{plainnat}
\bibliography{references}

@article{sun2021treatmentWP,
  title   = {Treatment Allocation under Uncertain Costs},
  author  = {Sun, Hao and Du, Shuyang and Wager, Stefan},
  journal = {arXiv preprint arXiv:2103.11066v1},
  year    = {2021},
  month   = {Mar},
  url     = {https://arxiv.org/abs/2103.11066v1}
}

@article{graham2012inverse,
  title={Inverse probability tilting for moment condition models with missing data},
  author={Graham, Bryan S and de Xavier Pinto, Cristine Campos and Egel, Daniel},
  journal={The Review of Economic Studies},
  volume={79},
  number={3},
  pages={1053--1079},
  year={2012},
  publisher={Oxford University Press}
}

@article{wooldridge2007inverse,
  title={Inverse probability weighted estimation for general missing data problems},
  author={Wooldridge, Jeffrey M},
  journal={Journal of econometrics},
  volume={141},
  number={2},
  pages={1281--1301},
  year={2007},
  publisher={Elsevier}
}

@article{manski1988ordinal,
  title={Ordinal utility models of decision making under uncertainty},
  author={Manski, Charles F},
  journal={Theory and Decision},
  volume={25},
  number={1},
  pages={79--104},
  year={1988},
  publisher={Springer}
}

@article{chamberlain1992efficiency,
  title={Efficiency bounds for semiparametric regression},
  author={Chamberlain, Gary},
  journal={Econometrica: Journal of the Econometric Society},
  pages={567--596},
  year={1992},
  publisher={JSTOR}
}

@article{dhaliwal2013comparative,
  title={Comparative cost-effectiveness analysis to inform policy in developing countries: a general framework with applications for education},
  author={Dhaliwal, Iqbal and Duflo, Esther and Glennerster, Rachel and Tulloch, Caitlin},
  journal={Education policy in developing countries},
  volume={17},
  pages={285--338},
  year={2013},
  publisher={University of Chicago Press Chicago}
}

@article{hendren2020unified,
  title={A unified welfare analysis of government policies},
  author={Hendren, Nathaniel and Sprung-Keyser, Ben},
  journal={The Quarterly Journal of Economics},
  volume={135},
  number={3},
  pages={1209--1318},
  year={2020},
  publisher={Oxford University Press}
}

@article{finkelstein2012oregon,
  title={The Oregon health insurance experiment: evidence from the first year},
  author={Finkelstein, Amy and Taubman, Sarah and Wright, Bill and Bernstein, Mira and Gruber, Jonathan and Newhouse, Joseph P and Allen, Heidi and Baicker, Katherine and Oregon Health Study Group},
  journal={The Quarterly journal of economics},
  volume={127},
  number={3},
  pages={1057--1106},
  year={2012},
  publisher={MIT Press}
}

@article{lemmens2020managing,
  title={Managing churn to maximize profits},
  author={Lemmens, Aur{\'e}lie and Gupta, Sunil},
  journal={Marketing Science},
  volume={39},
  number={5},
  pages={956--973},
  year={2020},
  publisher={INFORMS}
}

@article{bhattacharya2012inferring,
  title={Inferring welfare maximizing treatment assignment under budget constraints},
  author={Bhattacharya, Debopam and Dupas, Pascaline},
  journal={Journal of Econometrics},
  volume={167},
  number={1},
  pages={168--196},
  year={2012},
  publisher={Elsevier}
}

@article{yadlowsky2021evaluating,
  title={Evaluating Treatment Prioritization Rules via Rank-Weighted Average Treatment Effects},
  author={Yadlowsky, Steve and Fleming, Scott and Shah, Nigam and Brunskill, Emma and Wager, Stefan},
  journal={arXiv preprint arXiv:2111.07966},
  year={2021}
}

@article{chung2013exact,
  title={Exact and asymptotically robust permutation tests},
  author={Chung, EunYi and Romano, Joseph P},
  journal={The Annals of Statistics},
  volume={41},
  number={2},
  pages={484--507},
  year={2013},
  publisher={Institute of Mathematical Statistics}
}

@book{van2000asymptotic, place={Cambridge}, series={Cambridge Series in Statistical and Probabilistic Mathematics}, title={Asymptotic Statistics}, DOI={10.1017/CBO9780511802256}, publisher={Cambridge University Press}, author={Vaart, A. W. van der}, year={1998}, collection={Cambridge Series in Statistical and Probabilistic Mathematics}}

@article{newey1994large,
  title={Large sample estimation and hypothesis testing},
  author={Newey, Whitney K and McFadden, Daniel},
  journal={Handbook of econometrics},
  volume={4},
  pages={2111--2245},
  year={1994},
  publisher={Elsevier}
}

@article{sun2021empirical,
  title={Empirical Welfare Maximization with Constraints},
  author={Sun, Liyang},
  journal={arXiv preprint arXiv:2103.15298},
  year={2021}
}

@article{hoch2002something,
  title={Something old, something new, something borrowed, something blue: a framework for the marriage of health econometrics and cost-effectiveness analysis},
  author={Hoch, Jeffrey S and Briggs, Andrew H and Willan, Andrew R},
  journal={Health economics},
  volume={11},
  number={5},
  pages={415--430},
  year={2002},
  publisher={Wiley Online Library}
}

@article{xu2020estimating,
  title={Estimating the optimal individualized treatment rule from a cost-effectiveness perspective},
  author={Xu, Yizhe and Greene, Tom H and Bress, Adam P and Sauer, Brian C and Bellows, Brandon K and Zhang, Yue and Weintraub, William S and Moran, Andrew E and Shen, Jincheng},
  journal={Biometrics},
  year={2020},
  publisher={Wiley Online Library}
}

@article{rzepakowski2012decision,
  title={Decision trees for uplift modeling with single and multiple treatments},
  author={Rzepakowski, Piotr and Jaroszewicz, Szymon},
  journal={Knowledge and Information Systems},
  volume={32},
  number={2},
  pages={303--327},
  year={2012},
  publisher={Springer}
}

@article{gupta2020maximizing,
  title={Maximizing intervention effectiveness},
  author={Gupta, Vishal and Han, Brian Rongqing and Kim, Song-Hee and Paek, Hyung},
  journal={Management Science},
  year={2020},
  volume={forthcoming},
  publisher={INFORMS}
}

@article{imai2019experimental,
  title={Experimental evaluation of individualized treatment rules},
  author={Imai, Kosuke and Li, Michael Lingzhi},
  journal={arXiv preprint arXiv:1905.05389},
  year={2019}
}

@article{foster2019orthogonal,
  title={Orthogonal statistical learning},
  author={Foster, Dylan J and Syrgkanis, Vasilis},
  journal={arXiv preprint arXiv:1901.09036},
  year={2019}
}

@article{schick1986asymptotically,
  title={On asymptotically efficient estimation in semiparametric models},
  author={Schick, Anton},
  journal={The Annals of Statistics},
  pages={1139--1151},
  year={1986},
  publisher={JSTOR}
}

@article{wang2018instrumental,
  title={An instrumental variable forest approach for detecting heterogeneous treatment effects in observational studies},
  author={Wang, Guihua and Li, Jun and Hopp, Wallace J},
  journal={Management Science},
  volume={68},
  number={5},
  pages={3399--3418},
  year={2022},
  publisher={INFORMS}
}

@book{angrist2008mostly,
  author = {Angrist, Joshua D. and Pischke, J\"orn-Steffen},
  publisher = {Princeton University Press},
  shorttitle = {Mostly Harmless Econometrics},
  title = {Mostly Harmless Econometrics: An Empiricist's Companion},
  year = {2008}
}

@book{banerjee2011poor,
  title={Poor economics: A radical rethinking of the way to fight global poverty},
  author={Banerjee, Abhijit and Duflo, Esther},
  year={2011},
  publisher={Public Affairs}
}

@article{kohavi2009controlled,
  title={Controlled experiments on the web: survey and practical guide},
  author={Kohavi, Ron and Longbotham, Roger and Sommerfield, Dan and Henne, Randal M},
  journal={Data mining and knowledge discovery},
  volume={18},
  number={1},
  pages={140--181},
  year={2009},
  publisher={Springer}
}

@article{durbin1954errors,
  title={Errors in variables},
  author={Durbin, James},
  journal={Revue de l'Institut International de Statistique},
  pages={23--32},
  year={1954},
  publisher={JSTOR}
}

@article{dantzig1957discrete,
  title={Discrete-variable extremum problems},
  author={Dantzig, George B},
  journal={Operations Research},
  volume={5},
  number={2},
  pages={266--288},
  year={1957},
  publisher={INFORMS}
}

@article{basu2017detecting,
  title={Detecting heterogeneous treatment effects to guide personalized blood pressure treatment: A modeling study of randomized clinical trials},
  author={Basu, Sanjay and Sussman, Jeremy and Hayward, Rod},
  journal={Annals of Internal Medicine},
  volume={166},
  number={5},
  pages={354--360},
  year={2017},
  publisher={American College of Physicians}
}

@article{kennedy2020optimal,
  title={Optimal doubly robust estimation of heterogeneous causal effects},
  author={Kennedy, Edward H},
  journal={arXiv preprint arXiv:2004.14497},
  year={2020}
}

@article{bertsimas2019optimal,
  title={Optimal prescriptive trees},
  author={Bertsimas, Dimitris and Dunn, Jack and Mundru, Nishanth},
  journal={INFORMS Journal on Optimization},
  volume={1},
  number={2},
  pages={164--183},
  year={2019},
  publisher={INFORMS}
}

@article{luedtke2016optimal,
  title={Optimal individualized treatments in resource-limited settings},
  author={Luedtke, Alexander R and van der Laan, Mark J},
  journal={The International Journal of Biostatistics},
  volume={12},
  number={1},
  pages={283--303},
  year={2016},
  publisher={De Gruyter}
}

@article{ascarza2018retention,
  title={Retention futility: Targeting high-risk customers might be ineffective},
  author={Ascarza, Eva},
  journal={Journal of Marketing Research},
  volume={55},
  number={1},
  pages={80--98},
  year={2018},
  publisher={SAGE Publications Sage CA: Los Angeles, CA}
}

@article{breiman2001random,
  title={Random forests},
  author={Breiman, Leo},
  journal={Machine learning},
  volume={45},
  number={1},
  pages={5--32},
  year={2001},
  publisher={Springer}
}

@article{abadie2003semiparametric,
  title={Semiparametric instrumental variable estimation of treatment response models},
  author={Abadie, Alberto},
  journal={Journal of Econometrics},
  volume={113},
  number={2},
  pages={231--263},
  year={2003},
  publisher={Elsevier}
}

@article{angrist1996identification,
  title={Identification of causal effects using instrumental variables},
  author={Angrist, Joshua D and Imbens, Guido W and Rubin, Donald B},
  journal={Journal of the American Statistical Association},
  volume={91},
  number={434},
  pages={444--455},
  year={1996},
  publisher={Taylor \& Francis}
}

@article{athey2020policy,
  title={Policy learning with observational data},
  author={Athey, Susan and Wager, Stefan},
  journal={Econometrica},
  volume={89},
  number={1},
  pages={133--161},
  year={2021},
  publisher={Wiley Online Library}
}

@article{athey2019estimating,
  title={Estimating treatment effects with causal forests: An application},
  author={Athey, Susan and Wager, Stefan},
  journal={Observational Studies},
  volume={5},
  pages={36--51},
  year={2019}
}

@article{athey2019generalized,
  title={Generalized random forests},
  author={Athey, Susan and Tibshirani, Julie and Wager, Stefan},
  journal={The Annals of Statistics},
  volume={47},
  number={2},
  pages={1148--1178},
  year={2019},
  publisher={Institute of Mathematical Statistics}
}

@article{chernozhukov2018double,
  title={Double/debiased machine learning for treatment and structural parameters},
  author={Chernozhukov, Victor and Chetverikov, Denis and Demirer, Mert and Duflo, Esther and Hansen, Christian and Newey, Whitney and Robins, James},
  year={2018},
  volume={21},
  number={1},
  pages={1--68},
  journal={The Econometrics Journal},
  publisher={Oxford University Press Oxford, UK}
}

@book{efron1982jackknife,
  title={The Jackknife, the Bootstrap, and other Resampling Plans},
  author={Efron, Bradley},
  year={1982},
  publisher={Siam}
}

@article{hahn2020bayesian,
  title={Bayesian regression tree models for causal inference: regularization, confounding, and heterogeneous effects},
  author={Hahn, P Richard and Murray, Jared S and Carvalho, Carlos M},
  journal={Bayesian Analysis},
  year={2020},
  publisher={International Society for Bayesian Analysis}
}

@article{hirano2003efficient,
  title={Efficient estimation of average treatment effects using the estimated propensity score},
  author={Hirano, Keisuke and Imbens, Guido W and Ridder, Geert},
  journal={Econometrica},
  volume={71},
  number={4},
  pages={1161--1189},
  year={2003},
  publisher={Wiley Online Library}
}

@article{holland1986statistics,
  title={Statistics and causal inference},
  author={Holland, Paul W},
  journal={Journal of the American Statistical Association},
  volume={81},
  number={396},
  pages={945--960},
  year={1986},
  publisher={Taylor \& Francis}
}

@article{huang2020estimating,
  title={Estimating individualized treatment rules with risk constraint},
  author={Huang, Xinyang and Xu, Jin},
  journal={Biometrics},
  year={2020},
  publisher={Wiley Online Library}
}

@article{imbens1994identification,
  title={Identification and Estimation of Local Average Treatment Effects},
  author={Imbens, Guido W and Angrist, Joshua D},
  journal={Econometrica},
  volume={62},
  number={2},
  pages={467--475},
  year={1994}
}

@book{imbens2015causal,
  title={Causal Inference in Statistics, Social, and Biomedical Sciences},
  author={Imbens, Guido W and Rubin, Donald B},
  year={2015},
  publisher={Cambridge University Press}
}

@article{kallus2018confounding,
  title={Confounding-robust policy improvement},
  author={Kallus, Nathan and Zhou, Angela},
  journal={Management Science},
  volume={forthcoming},
  year={2020}
}

@article{kitagawa2018should,
  title={Who should be treated? Empirical welfare maximization methods for treatment choice},
  author={Kitagawa, Toru and Tetenov, Aleksey},
  journal={Econometrica},
  volume={86},
  number={2},
  pages={591--616},
  year={2018},
  publisher={Wiley Online Library}
}

@article{kunzel2019metalearners,
  title={Metalearners for estimating heterogeneous treatment effects using machine learning},
  author={K{\"u}nzel, S{\"o}ren R and Sekhon, Jasjeet S and Bickel, Peter J and Yu, Bin},
  journal={Proceedings of the National Academy of Sciences},
  volume={116},
  number={10},
  pages={4156--4165},
  year={2019},
  publisher={National Acad Sciences}
}

@article{nie2020quasi,
  title={Quasi-oracle estimation of heterogeneous treatment effects},
  author={Nie, Xinkun and Wager, Stefan},
  journal={Biometrika},
  volume={108},
  number={2},
  pages={299--319},
  year={2021},
  publisher={Oxford University Press}
}

@article{robinson1988root,
  title={Root-N-consistent semiparametric regression},
  author={Robinson, Peter M},
  journal={Econometrica},
  volume={56},
  number={4},
  pages={931--954},
  year={1988},
  publisher={JSTOR}
}

@article{rosenbaum1983central,
  title={The central role of the propensity score in observational studies for causal effects},
  author={Rosenbaum, Paul R and Rubin, Donald B},
  journal={Biometrika},
  volume={70},
  number={1},
  pages={41--55},
  year={1983},
  publisher={Oxford University Press}
}

@article{stoye2009minimax,
  title={Minimax regret treatment choice with finite samples},
  author={Stoye, J{\"o}rg},
  journal={Journal of Econometrics},
  volume={151},
  number={1},
  pages={70--81},
  year={2009},
  publisher={Elsevier}
}

@article{wager2018estimation,
  title={Estimation and inference of heterogeneous treatment effects using random forests},
  author={Wager, Stefan and Athey, Susan},
  journal={Journal of the American Statistical Association},
  volume={113},
  number={523},
  pages={1228--1242},
  year={2018},
  publisher={Taylor \& Francis}
}

@article{wang2018learning,
  title={Learning optimal personalized treatment rules in consideration of benefit and risk: with an application to treating type 2 diabetes patients with insulin therapies},
  author={Wang, Yuanjia and Fu, Haoda and Zeng, Donglin},
  journal={Journal of the American Statistical Association},
  volume={113},
  number={521},
  pages={1--13},
  year={2018},
  publisher={Taylor \& Francis}
}

@article{zhao2012estimating,
  title={Estimating individualized treatment rules using outcome weighted learning},
  author={Zhao, Yingqi and Zeng, Donglin and Rush, A John and Kosorok, Michael R},
  journal={Journal of the American Statistical Association},
  volume={107},
  number={499},
  pages={1106--1118},
  year={2012},
  publisher={Taylor \& Francis}
}

\newpage

\appendix

\section{Proofs} \label{appendix}
\subsection*{Proof of Theorem \ref{theo:pi_opt}} \label{app:theo1}

To ease the presentation, we first define the conditional average treatment effect function for both
rewards and costs as
$$\tau_C(x) = \mathbb E[C_i(1) - C_i(0) \cond X_i=x], \quad \tau_Y(x) = \mathbb E[Y_i(1) - Y_i(0) \cond X_i=x].$$
Because $C_i(1) \geq C_i(0)$ almost surely, we see that $G(\rho) = E[I\{\rho(X_i) > \rho\}\tau_C(X_i)]$ is a non-increasing function of $\rho$. Let
$$\eta_B := \inf \{\rho: G(\rho) \le B\}, \ \ \ \ \rho_B = \max\{\eta_B, \, 0\}.$$
The claimed optimal (stochastic) decision rule in  \eqref{eq:pi_opt} can then be rewritten as
\begin{equation}  \label{app:theo1-pi-stochastic}
\pi^*_B(x) =
\begin{cases}
a_B  & \text{if }  \ \rho(x) = \rho_B, \\
1 & \text{if } \ \rho(x) > \rho_B,
\end{cases}
\end{equation}
where
\begin{equation} \label{app:theo1-aB}
a_B =
\begin{cases}
0  & \text{if }  \ E[I\{\rho(X_i) = \rho_B\}\tau_C(X_i)] = 0, \\
\min\left\{\frac{B - E\left[ I\{\rho(X_i) > \rho_B\} \tau_C(X_i) \right]}{E[I\{\rho(X_i) = \rho_B\}\tau_C(X_i)]}, \, 1\right\} & \text{if } \ E[I\{\rho(X_i) = \rho_B\}\tau_C(X_i)] > 0.
\end{cases}
\end{equation}
Note that $\pi^*_B(x)$ and $I\{\rho(x) > \rho_B\}$ are almost surely equal if $\PP{\rho(X_i) = \rho_B} = 0$ or if $\eta_B < 0$, and they should return the same decision in these settings. Moreover, $E[\pi^*_B(X_i)\tau_C(X_i)] = B$ if $\rho_B > 0$.

To verify that the above rule is in fact optimal, let
$r(X_i)$ denote any other stochastic treatment rule which satisfies the budget constraint $B$. It remains to argue that
$$E[\tau_Y(X_i) \pi_B^*(X_i)] \geq E[\tau_Y(X_i) r(X_i)],$$
i.e., that $r(X_i)$ cannot achieve higher rewards than $\pi^*_B$ while respecting the budget. From now on, we assume that $\tau_C(X_i) > 0$ almost surely, i.e., that there are no units that are free to treat in expectation; because if there are units with $\tau_C(X_i) = 0$ then clearly one should just treat them according to the sign of $\tau_Y(X_i)$ (as is done by our policy), and this has no budget implications. Given this setting, we see that
\begin{equation}
\label{eq:LB}
\begin{split}
E[\tau_Y(X_i) (\pi_B^*(X_i) -  r(X_i))]
&= E[ \rho(X_i) \tau_C(X_i) (\pi_B^*(X_i) -  r(X_i))] \\
&\geq \rho_B E[ \tau_C(X_i) (\pi_B^*(X_i) -  r(X_i))],
\end{split}
\end{equation}
where the inequality follows by observing that, by definition of
$\pi_B^*$, we must have $\pi_B^*(X_i) -  r(X_i) \geq 0$ whenever $\rho(X_i) > \rho_B$ and $\pi_B^*(X_i) -  r(X_i) \leq 0$ whenever $\rho(X_i) < \rho_B$.

We conclude by considering two cases: Either $\rho_B > 0$ or $\rho_B = 0$. In the first case, we know that $\pi^*_B$ spends the whole budget, i.e., $E[ \tau_C(X_i) \pi_B^*(X_i)] = B$; thus, by the budget constraint on $r(X_i)$ (i.e., $E[ \tau_C(X_i) r(X_i)] \leq B$), we see that $E[\tau_Y(X_i) (\pi_B^*(X_i) -  r(X_i))] \geq 0$. Meanwhile, in the second case, the lower bound in \eqref{eq:LB} is 0, and so our conclusion again holds.
Finally, by an extension of the same argument, we see that when $\PP{\rho(X_i) = \rho_B} = 0$, our policy $\pi^*_B(x)$ is  almost surely equivalent to $I\{\rho(x) > \rho_B\}$, and is both deterministic and the unique reward-maximizing decision rule that respects the budget constraint.

\subsection*{Proof of Proposition \ref{prop:unconf}}\label{app:propo2}

In this section, we show the equation in Proposition \ref{prop:unconf}. Assume $W_i \in \{0, 1\}$ and let $h_w(x) := \PP{W_i = 1 \cond x}$. Notice that
\begin{align} \label{app:prop2-cov}
    \begin{split}
        & \Cov{Y_i, \, W_i \cond X_i=x} \\
        =& E[Y_iW_i \cond X_i=x] - E[Y_i \cond X_i = x]E[W_i \cond X_i=x] \\
        =& E[Y_i(1)W_i \cond X_i =x] - E[Y_i \cond X_i = x]E[W_i \cond X_i =x] \\
        =& h_w(x)E[Y_i(1) \cond X_i =x] - h_w(x)^2E[Y_i(1) \cond X_i =x] - h_w(x)\{1-h_w(x)\}E[Y_i(0) \cond X_i = x] \\
        =& h_w(x)\{1-h_w(x)\}\{E[Y_i(1) \cond X_i = x] - E[Y_i(0) \cond X_i = x]\} \\
        =& h_w(x)\{1-h_w(x)\}\tau_Y(x),
    \end{split}
\end{align}
where the second equality comes from the consistency assumption that $Y_i = W_iY_i(1) + (1-W_i)Y_i(0)$ and the third equality comes from the unconfoundedness assumption in Proposition \ref{prop:unconf}. Similarly, we can show that
\begin{equation*}
    \Cov{C_i, \, W_i \cond X_i = x} = h_w(x)\{1-h_w(x)\}\tau_C(x)
\end{equation*}
and thus
\begin{align} \label{app:prop2-rho}
    \begin{split}
        & \frac{\Cov{Y_i, \, W_i \cond X_i = x}}{\Cov{C_i, \, W_i \cond X_i = x}} \\
        =& \frac{h_w(x)\{1-h_w(x)\}\tau_Y(x)}{h_w(x)\{1-h_w(x)\}\tau_C(x)} \\
        =& \frac{\tau_Y(x)}{\tau_C(x)} \\
        =& \rho(x),
    \end{split}
\end{align}
which completes the proof of the Proposition \ref{prop:unconf}.

\subsection*{Proof of Theorem \ref{theo:rate}}

Recall that we are in a decision rule where all policies are non-randomized, i.e.,
$\pi^*(x) = 1[\rho(x) > \rho^*_B]$ and $\hat{\pi}(x) = 1[\hrho(x) > \hrho_B]$. Now,
\begin{equation*}
\begin{aligned}
V(\pi^*_B) - V(\hat{\pi}_B) &= \mathbb{E}_T[\tau_Y(X_i)(\pi^*_B(X_i) - \hat{\pi}_B(X_i))] \\
&= \mathbb{E}_T[\rho(X_i)\tau_C(X_i)(\pi^*_B(X_i) - \hat{\pi}_B(X_i))] \\
&= \hrho_B \mathbb{E}_T[\tau_C(X_i)(\pi^*_B(X_i) - \hat{\pi}_B(X_i))] \\
&\quad\quad\quad\quad + \mathbb{E}_T[\tau_C(X_i)(\rho(X_i) - \hrho_B)(\pi^*_B(X_i) - \hat{\pi}_B(X_i))] \\
&= \hrho_B \delta + \mathbb{E}_T[\tau_C(X_i) (\rho(X_i) - \hrho_B)(\pi^*_B(X_i) - \hat{\pi}_B(X_i))].
\end{aligned}
\end{equation*}
We study the term $(\rho(X_i) - \hrho_B)(\pi^*_B(X_i) - \hat{\pi}_B(X_i))$ by cases. If
$\pi^*(x) = 1$ and $\hat{\pi}_B(x) = 0$, then $\hrho(x) \leq \hrho_B$ and so
$$  (\rho(X_i) - \hrho_B)(\pi^*_B(X_i) - \hat{\pi}_B(X_i)) \leq  (\rho(X_i) - \hrho(x))(\pi^*_B(X_i) - \hat{\pi}_B(X_i)). $$
If $\pi^*(x) = 0$ and $\hat{\pi}_B(x) = 1$, then $\hrho(x) > \hrho_B$ and now $\pi^*_B(X_i) - \hat{\pi}_B(X_i) = -1$, so
$$  (\rho(X_i) - \hrho_B)(\pi^*_B(X_i) - \hat{\pi}_B(X_i)) \leq  (\rho(X_i) - \hrho(x))(\pi^*_B(X_i) - \hat{\pi}_B(X_i)). $$
Finally, if $\pi^*(x) = \hat{\pi}_B(x)$ the term is obviously 0. We thus conclude
\begin{align*}
V(\pi^*_B) - V(\hat{\pi}_B)
&\leq \mathbb{E}_T[\tau_C(X_i) (\rho(X_i) - \hrho(x))(\pi^*_B(X_i) - \hat{\pi}_B(X_i))] + \hrho_B \delta \\
&\leq \mathbb{E}_T[\tau_C(X_i) |\rho(X_i) - \hrho(x)|] + \hrho_B \delta,
\end{align*}
recalling that $\tau_C(x) \geq 0$ by assumption.

\subsection*{Derivation of Equation \ref{eq:conditional_mom}}

When $\rho(x) = x' \beta$, then the equation in Proposition \ref{prop:unconf} is equivalent to
\[ x'\beta = \frac{\Cov{Y_i, \, W_i \cond X_i = x}}{\Cov{C_i, \, W_i \cond X_i = x}} \]
Using the definition of conditional covariance, and rearranging, we have that
\begin{align*}  &  \mathbb E[ (W_i - h_w(X_i)) (C_i - h_c(X_i) )  | X_i =x]  x'\beta  = \mathbb E[ (W_i - h_w(X_i) ) (Y_i - h_y(X_i)) | X_i = x]  \\
& 0  =    \mathbb E[ (W_i - h_w(X_i) ) (Y_i - h_y(X_i)) | X_i = x]  - \mathbb E[  (W_i - h_w(X_i)) (C_i - h_c(X_i) )  X_i' \beta | X_i =x] \\
 & 0  = \mathbb E[(W_i - h_w(X_i)) ((Y_i - h_y(X_i)) - (C_i - h_c(X_i)) X_i' \beta) | X_i = x]
\end{align*}
This is equivalent to
\[ \mathbb E[ e_i(\beta, h(X_i)) | X_i = x] = 0 \]

\subsection*{Proof of Theorem \ref{thm:dml}}

We will use Theorem 3.1 and Theorem 3.2 of \citet{chernozhukov2018double}, therefore we need to verify the Assumptions 3.1 and 3.2 from the paper, which will complete the proof.

We have $\beta$ that satisfies the following unconditional moment restriction
\[ \mathbb E[\psi_i(\beta, h(X_i))] = 0, \]
where the score function is 
\begin{align}
\psi_i(\beta, h(X_i))&= X_iV_iU_i = X_i[W_i -h_w(X_i)][Y_i - h_y(X_i) -(C_i - h_c(X_i)) X_i' \beta] \notag\\&= \psi_i^{(0)}(h(X_i)) + \psi_i^{(1)}(h(X_i))\beta. \label{eq:score}
\end{align}

The score function is linear in $\beta$. This verifies Assumption 3.1b) of \citet{chernozhukov2018double}. To apply the Theorem, we  must verify the remaining components of Assumption 3.1 and Assumption 3.2.

3.1a) is satisfied, since the equation in Proposition \ref{prop:unconf} is equivalent to $\mathbb E[\psi_i(\beta, h(X_i))] = 0$ under the linearity assumption for $\rho(x)$.

3.1c) is satisfied, since the score function is linear in both $\beta$ and the nuisance parameters, it is twice differentiable in the nuisance parameters. For 3.1d), we show Neyman-Orthogonality by showing that the partial derivative, evaluated at zero, of the conditional moment restriction with respect to each component of a perturbation of the nuisance functions is zero. Then, the Law of Iterated Expectations implies Neyman-Orthogonality for the unconditional score function.
\begin{align*}
\frac{\partial \mathbb E[ e_i (\beta, h(x) + \epsilon \delta(x) )| X_i = x] }{ \partial \epsilon_y } \Big \rvert_{\epsilon = 0} &=  - \delta_y(x) \mathbb E[W_i - \mathbb E[W_i | X_i = x] | X_i  = x ]
\\ &= 0
\end{align*}
Similarly,
\begin{align*}
\frac{\partial \mathbb E[ e_i (\beta, h(x) + \epsilon \delta(x) )| X_i = x] }{ \partial \epsilon_c } \Big \rvert_{\epsilon = 0} &= X^\prime \beta \delta_c(x) \mathbb E[W_i - \mathbb E[W_i | X_i = x] | X_i  = x]
\\ &=  0.
\end{align*}
Lastly,
\begin{align*}
\frac{\partial \mathbb E[ e_i (\beta, h(x) + \epsilon \delta(x) )| X_i = x] }{ \partial \epsilon_w } =&  -\delta_w(x) \mathbb E[Y_i - \mathbb E[Y_i | X_i= x] | X_i = x] \\&+ X^\prime \beta \delta_w(x) \mathbb E[C_i - \mathbb E[C_i | X_i = x] | X_i  = x]
\\ ={}& 0.
\end{align*}
For 3.1e), we need that  $\mathbb E[V_i D_i X_iX_i']$ is invertible and
\begin{align*}
\mathbb E[ V_i D_i X_i X_i'] & = \mathbb E[(W_i - \mathbb E[W_i | X_i])( C_i - \mathbb E[C_i | X_i])X_i X_i'] \\
& = \mathbb E_x \left[\mathbb E\left[(W_i - \mathbb E[W_i \mid X_i])( C_i - \mathbb E[C_i \mid X_i])X_iX_i' \mid X_i\right]\right]\\
& = \mathbb E_x[X_i X_i' \mbox{Cov}(C_i, W_i | X_i) ].
\end{align*}

Since we are in the setting of the Proposition \ref{prop:unconf} and unconfoundedness applies as well as the overlap condition $0 < h_w(x) <1$, we have that
\[ \mbox{Cov}(C_i, W_i | X_i) = h_w(x) ( 1- h_w(x)) \mathbb E[C_i (1) - C_i(0) | X_i = x] > 0, \]
where the inequality is from Assumption \ref{as:c}. Then, $\mathbb E[V_i D_i']$ is invertible as long as $\mathbb E[X_i X_i']$ is full rank, which is by assumption, further, the singular values of $X_i X_i'$ are bounded from above, since $X_i$ are bounded.

We now verify the assumptions 3.2a) through c). The point is to show various bounds on $\psi^{(1)}$ and $\psi$ defined in \eqref{eq:score} with constants $a, A$ and a sequence $\delta_n$, featuring in the Assumptions \ref{as:c}, \ref{assu:nuisance_rates}, \ref{assu:regularity}.

We will first introduce and discuss additional notation: \cite{chernozhukov2018double} Assumption 3.2 requires bounds involving some vector or matrix norm $\|A\|$, which we chose to be an $\infty$ norm: $\|A\|_\infty = \max_{ij} |A_{ij}|$. For our own purposes we will also use the $q$ norm of a random scalar denoted as $\|\xi\|_{P,q} = \mathbb E_T[|\xi|^q]^{1/q}$. Also $\mathbf{1}_m$ means a column vector of $1$ of a size m.

\paragraph{Useful Inequalities.}

Before turning to the verification of the assumptions, we will derive some useful bounds, which are used throughout the proof.

Throughout all of the derivations we will use the following inequalities (for any $p < q$), which hold by Assumptions \ref{assu:nuisance_rates} and \ref{assu:regularity}:
\begin{equation}
   \begin{gathered}
\|W_i - \hat h_w(X_i)\|_{P, p} \leq \|W_i - \hat h_w(X_i)\|_{P, q} \leq \|h_w(X_i) - \hat h_w(X_i)\|_{P, q} + \|W_i\|_{P, q} + \|h_w(X_i)\|_{P, q} \leq 3A \\
\|C_i - \hat h_c(X_i)\|_{P, p} \leq \|C_i - \hat h_c(X_i)\|_{P, q} \leq \|h_c(X_i) - \hat h_c(X_i)\|_{P, q} + \|C_i\|_{P, q} + \|h_c(X_i)\|_{P, q} \leq 3A \\
\|Y_i - \hat h_y(X_i)\|_{P, p} \leq \|Y_i - \hat h_y(X_i)\|_{P, q} \leq \|h_y(X_i) - \hat h_y(X_i)\|_{P, q} + \|Y_i\|_{P, q} + \|h_y(X_i)\|_{P, q} \leq 3A. \label{eq:norm_bounds}
\end{gathered} 
\end{equation}
We can replace the RHS by $2A$ if we have a population version of $h(X_i)$ on the left hand side, which will be useful for derivation of a bound on $\beta$.

\paragraph{Bounding $\beta$}

Let us first bound $\mathbb{E}[(W_i - h_w(X_i))(C_i - h_c(X_i))X_iX_i']$. 
Under Assumption \ref{assu:unconf}: $\mathbb{E}[(W_i - h_w(X_i))(C_i - h_c(X_i))X_iX_i'] = \mathbb{E}[X_iX_i'\mathbb{E}[(W_i - h_w(X_i))(C_i - h_c(X_i)) \mid X_i]]  = \mathbb{E}[X_iX_i'\mathbb{E}[C_i(1) - C_i(0) \mid X_i]]  = \mathbb{E}[X_iX_i'\tau_C(X_i)]$.

We use the assumption that the matrix $\EE{X_iX_i'}$ is invertible, therefore it's singular values are bounded from below. Let us also assume that the constant $a$ is low enough so that $2a$ is a valid bound for singular values of $\mathbb{E}[X_iX_i']$. Therefore $\mathbb{E}[X_iX_i'] \geq 2a I$

Let us define $S(c) = \mathbb{E}[X_iX_i'  1\{\tau_C(X_i) >
c\}]$, and also define $\bar c = \sup\{c : S(c) \geq a I\}$

By Assumption \ref{as:c} $P(\tau_C(X_i) = 0) = 0$, therefore $S(0) \geq 2a I$. Also $S(0)$ is right-continuous, therefore $\bar{c} > 0$. To save on notation let us also assume that $a$ is small enough that $\bar{c} \geq a$.

$\mathbb{E}[\tau_C(x) X_i X_i'] \geq \mathbb{E}[\tau_C(x) X_i X_i' 1\{\tau_C(x) >
a\}] \geq a \mathbb{E}[X_i X_i'  1\{\tau_C(x) >
a\}] \geq a \bar c I \geq a^2 I$

Therefore we established a lower bound on $\mathbb{E}[(W_i - h_w(X_i))(C_i - h_c(X_i))X_i X_i']$, now we can derive an upper bound on $\beta$:

\begin{gather}
\begin{aligned}
\beta & = \mathbb{E}(X_i X_i'(W_i - h_w(X_i))(C_i - h_c(X_i)))^{-1}\mathbb{E}(X_i'(W_i - h_w(X_i))(Y_i - h_y(X_i))) \\
& \leq a^{-2}\mathbf{1}_m A \|W_i - h_w(X_i)\|_{P, 2} \|Y_i - h_y(X_i)\|_{P, 2} \\
& \leq a^{-2}\mathbf{1}_m A \|W_i - h_w(X_i)\|_{P, q} \|Y_i - h_y(X_i)\|_{P, q} \\
& \leq \mathbf{1}_m 4 a^{-2} A^2. \label{eq:beta}
\end{aligned}
\end{gather}

Further, we will use the bound on $X_i$. This will give the following related bounds:
\begin{equation*}
  \|X_i\|_{\infty} \le A 
  \qquad 
  \|X_i X_i'\|_{\infty} \le A^2 
  \qquad 
  \|X_i X_i' \beta\|_{\infty} \le 4 a^{-2} m A^4
\end{equation*}

\paragraph{Verifying assumption 3.2a) of \cite{chernozhukov2018double}} Let the realization set $\mathcal{T}_N$ be the set of estimates satisfying the conditions in the Assumption \ref{assu:nuisance_rates}. Establishing the bounds below we will consider $\hat h(X_i)$ functions from this realization set.

We use the Holder inequality, a bound on $\|X_iX_i'\|_{\infty}$ and the previously derived bounds to derive a bound on $(E_T[\|\psi^{(1)}(\hat h(X_i))\|_\infty^{q / 2}])^{2 / q}$:
\begin{equation*}
\begin{aligned}
\left(E_T\left[\left\|\psi^{(1)}(\hat h(X_i))\right\|_\infty^{q / 2}\right]\right)^{2 / q} & =\|X_iX_i'(C_i - \hat h_c(X_i))(W_i-\hat h_w(X_i))\|_{P, q / 2} \\
&\leq \|X_iX_i'\|_{\infty}\|C_i - \hat h_c(X_i)\|_{P, q}\|W_i-\hat h_w(X_i)\|_{P, q} \leq A^2 9A^2 = 9A^4.
\end{aligned}
\end{equation*}

Now we will reuse the bound above to verify the second equation of Assumption 3.2b) of \cite{chernozhukov2018double}. We also use the established bound on $\beta$ \eqref{eq:beta}:
\begin{equation*}
\begin{aligned}
(E[\|\psi(\beta,\hat h(X_i))\|_\infty^{q / 2}])^{2 / q} ={}& \|\psi(\beta, \hat h(X_i))\|_{P, q / 2} \\
 ={}& \|X_i(W_i-\hat h_w(X_i))(Y_i - \hat h_y(X_i) - X_i'(C_i - \hat h_c(X_i)) \beta)\|_{P, q / 2} \\
 \leq{}& \|X_i(Y_i - \hat h_y(X_i))(W_i-\hat h_w(X_i))\|_{P, q / 2} \\&+ \|X_iX_i'\beta(C_i - \hat h_c(X_i))(W_i-\hat h_w(X_i))\|_{P, q / 2} \\
 \leq{}& A\|Y_i - \hat h_y(X_i)\|_{P, q}\|W_i - \hat h_w(X_i)\|_{P, q} \\&+ 4a^{-2}mA^4\|C_i - \hat h_c(X_i)\|_{P, q}\|W_i-\hat h_w(X_i)\|_{P, q} \\
 \leq{}& 9A^3 + 36a^{-2}mA^6.
\end{aligned}
\end{equation*}

Therefore we established an upper bound on $(E[\|\psi(\beta,\hat h(X_i))\|_\infty^{q / 2}])^{2 / q}$ and $(E[\|\psi^{(1)}(\hat h(X_i))\|_\infty^{q / 2}])^{2 / q}$ as required by the assumption.

\paragraph{Verifying Assumptions 3.2c) of \cite{chernozhukov2018double}}. Here we need to show the convergence to 0 of $\|E_T[\psi^{(1)}(\hat h(X_i))]-E_T[\psi^{(1)}(h(X_i))]\|_\infty$, $(E_T[\|\psi(\beta, \hat h(X_i))-\psi(\beta, h(X_i))\|_\infty^2])^{1 / 2}$ and $\sqrt{n}\|\partial_r^2 E_T[\psi(\beta, h(X_i) + r(\hat h(X_i)-h(X_i)))]\|_\infty$.

For the first equation we use boundedness of $X$, Assumption \ref{assu:nuisance_rates} and the bounds \eqref{eq:norm_bounds}:
\begin{equation*}
\begin{aligned}
\|E_T&[\psi^{(1)}(\hat h(X_i))]-E_T[\psi^{(1)}(h(X_i))]\|_\infty
\\={}& \|E_T[(W_i - \hat h_w(X_i))(C_i - \hat h_c(X_i))X_iX_i' - (W_i - h_w(X_i))(C_i - h_c(X_i))X_iX_i']\|_\infty \\
\leq {}& A^2\|(W_i - \hat h_w(X_i))(C_i - \hat h_c(X_i)) - (W_i - h_w(X_i))(C_i - h_c(X_i))\|_{P, 1} \\
\leq {}& A^2\|(h_w(X_i) - \hat h_w(X_i))(C_i - h_c(X_i))\|_{P, 1} \\
& + A^2\|(W_i - h_w(X_i))(h_c(X_i) - \hat h_c(X_i))\|_{P, 1} \\
&  + A^2\|(h_w(X_i) - \hat h_w(X_i))(h_c(X_i) - \hat h_c(X_i))\|_{P, 1} \\
\leq {}& A^2\|h_w(X_i) - \hat h_w(X_i)\|_{P, 2}\|C_i - h_c(X_i)\|_{P, 2} \\
& + A^2\|W_i - h_w(X_i)\|_{P, 2}\|h_c(X_i) - \hat h_c(X_i)\|_{P, 2} \\
& + A^2\|h_w(X_i) - \hat h_w(X_i)\|_{P, 2}\|h_c(X_i) - \hat h_c(X_i)\|_{P, 2} \\
\leq {}& 4A^3 \delta_n + A^3 \delta_n/\sqrt{n}.
\end{aligned}
\end{equation*}

Deriving the next inequality, we use the boundedness of conditional variance of $U_i$, the fact that $\|\hat h_w(X_i) - h_w(X_i)\|_\infty$ and $\|V_i\|_\infty$ are less than 1 (both $\hat h_w(X_i)$, $h_w(X_i)$ map into $[0, 1]$, $W \in \{0, 1\}$), and the bounds on $X_i$ and $\beta$:
\begin{equation*}
\begin{aligned}
(E_T[\|\psi&(\beta, \hat h(X_i))-\psi(\beta, h(X_i))\|_\infty^2])^{1 / 2}= \\
={}&\|X_i(V_i + h_w(X_i) - \hat h_w(X_i))(U_i + h_y(X_i) - \hat h_y(X_i) + (h_c(X_i) - \hat h_c(X_i))X_i'\beta) - V_iU_i\|_{P, 2} \\
\leq{}& A\|(h_w(X_i)-\hat h_w(X_i))U_i\|_{P, 2} + A\|(h_y(X_i) - \hat h_y(X_i))V_i\|_{P, 2} + 4a^{-2}mA^4\|(h_c(X_i) - \hat h_c(X_i))V_i\|_{P, 2} \\
& + A\|(h_y(X_i) - \hat h_y(X_i))(h_w(X_i)-\hat h_w(X_i))\|_{P, 2} + 4a^{-2}mA^4\|(h_c(X_i) - \hat h_c(X_i)) (h_w(X_i)-\hat h_w(X_i))\|_{P, 2}\\
\leq{} & A\|h_w(X_i)-\hat h_w(X_i)\|_{P, 2} + A^2\|\hat h_y(X_i) - h_y(X_i)\|_{P, 2} + 4a^{-2}mA^5\|(\hat h_c(X_i) - h_c(X_i))\|_{P, 2}\\
& + A\|(h_y(X_i) - \hat h_y(X_i))\|_{P, 2} \\&+ 4a^{-2}mA^4\|(h_c(X_i) - \hat h_c(X_i))\|_{P, 2} \\
\leq {} & (2A + A^{2} + 4a^{-2}mA^5 + 4a^{-2}mA^4)\delta_n\\
\end{aligned}
\end{equation*}

Finally, let
\begin{equation*}
  f(r) = \EE[T]{X_i(U_i - r(\hat h_y(X_i) - h_y(X_i)) + r(\hat h_c(X_i) - h_c(X_i))X_i'\beta)(V_i - r(\hat h_w(X_i) - h_w(X_i))}.
\end{equation*}
The derivative: 
\begin{gather*}
\begin{aligned}
  \partial f(r) ={}& \EE[T]{X_i(\hat h_y(X_i) - h_y(X_i))(V_i - r(\hat h_w(X_i) - h_w(X_i))} \\
  &+ \EE[T]{X_i(\hat h_c(X_i) - h_c(X_i))X_i'\beta(V_i - r(\hat h_w(X_i) - h_w(X_i))} \\
  &- \EE[T]{X_i(U_i - r(\hat h_y(X_i) - h_y(X_i)) + r(\hat h_c(X_i) - h_c(X_i))X_i'\beta)(\hat h_w(X_i) - h_w(X_i))}
\end{aligned}
\end{gather*}
\begin{gather*}
  \partial^2 f(r) = 2\EE[T]{X_i((\hat h_y(X_i) - h_y(X_i)) - (\hat h_c(X_i) - h_c(X_i))X_i'\beta)(\hat h_w(X_i) - h_w(X_i))}
\end{gather*}

We can bound
\begin{equation*}
\begin{aligned}
  | \partial^2 f(r) | &\leq 2 \|X_i(\hat h_y(X_i) - h_y(X_i))(\hat h_w(X_i) - h_w(X_i)\|_{P, 1} \\&+ 2\|X_iX_i'\beta(\hat h_c(X_i) - h_c(X_i))(\hat h_w(X_i) - h_w(X_i))\|_{P, 1} \\
  &\leq 2A\delta_n/\sqrt{n} + 8a^{-2}mA^4\delta_n/\sqrt{n}
\end{aligned}
\end{equation*}

This establishes the convergence to 0 of $\|E_T[\psi^{(1)}(\hat h(X_i)) - E_T[\psi^{(1)}(h(X_i))]\|_\infty$, $(E_T[\|\psi(\beta, \hat h(X_i)) - \psi(\beta, h(X_i))\|_\infty^2])^{1 / 2}$ and $\sqrt{n}\|\partial_r^2 E_T[\psi(\beta, h(X_i) + r(\hat h(X_i) - h(X_i)))]\|_\infty$

Assumption 3.2 d) also requires that the variance of the score $\mathbb E[V^2_iU^2_iX_iX_i']$ is non-degenerate. $\EE{V^2_iU^2_iX_iX_i'} = \EE{\EE{V^2_iU^2_i \cond X_i }X_iX_i'} \geq a
\EE{X_iX_i'}$, which is full rank by assumption.

Given we have verified that Assumptions 3.1 and 3.2 hold, then the result of Theorem \ref{thm:dml} comes directly from Theorem 3.1 of \citet{chernozhukov2018double}.

\subsection*{Proof of Theorem \ref{thm:grf}}

To establish the result of Theorem \ref{thm:grf}, we will use Theorem 5 of \citet{athey2019generalized}. The procedure described in Section \ref{sec:nonparametric} applies the generalized random forest algorithm to the conditional moment restrictions \eqref{eq:conditional_mom} on parameters $(\rho, h_w, h_y, h_c)$. To simplify the proof, we will reparametrize the moment restrictions. We will replace $h_y$ and $h_c$ with a single nuisance parameter: $h(x) = h_y(x) - \rho(x) h_c(x)$, so the conditional moment restrictions reduce to:
\begin{align*}
M^1_{\rho, h}(x) &= \mathbb E[\psi^1_{\rho, h}(W_i, Y_i, C_i) | X_i = x] = 0 \\
M^2_{\rho, h}(x) &= \mathbb E[\psi^2_{\rho, h}(W_i, Y_i, C_i) | X_i = x] = 0
\end{align*}
with
\begin{align*}
\psi^1_{\rho, h}(W_i, Y_i, C_i) &= W_i(Y_i - \rho C_i - h) \\
\psi^2_{\rho, h}(W_i, Y_i, C_i) &= Y_i - \rho C_i - h
\end{align*}

We need to verify the Assumptions 1--6 of \citet{athey2019generalized}, which will complete the proof. The requirement $\operatorname{Var}[\rho(X_i) \mid X_i = x] > 0$ stated in the text of Theorem 5 of \citet{athey2019generalized} is satisfied by Assumption \ref{assu:regularity}.

\paragraph{Verification of Assumption 1} 

The assumption states that the moments should be Lipschitz continuous in $x$. We can write both moments as linear combinations of conditional expectations and nuisance functions that are Lipschitz by Assumption \ref{as:lipschitz}:
\begin{align*}
  \label{eq:mom_for_grf}
  M^1_{\rho, h}(x) &= \mathbb E[W_i(Y_i - \rho(x) C_i - h(x)) \mid X_i = x] \\
  &= \mathbb E[W_iY_i \mid X_i = x] - \rho(x)\,\mathbb E[W_iC_i \mid X_i = x] - h(x)\,\mathbb E[W_i \mid X_i = x], \\
  M^2_{\rho, h}(x) &= \mathbb E[Y_i - \rho(x) C_i - h(x) \mid X_i = x] \\
  &= \mathbb E[Y_i \mid X_i = x] - \rho(x)\,\mathbb E[C_i \mid X_i = x] - h(x).
\end{align*}
Therefore, both moments are Lipschitz in $x$.

\paragraph{Verification of Assumption 2}

The assumption requires that for fixed $x$, the moment function $M(x, \theta, \nu)$ is twice continuously differentiable in $(\theta, \nu)$ with a uniformly bounded second derivative, and that $V(x) := \partial_{(\theta, \nu)} M(x, \theta, \nu)$ is invertible at the true parameters.

The moments are affine in $(\rho, h)$, hence twice continuously differentiable with zero second derivatives. The Jacobian with respect to $(\rho, h)$ at the true parameters is
\begin{equation}
\label{eq:jacobian_grf}
V(x)
= \partial_{(\rho, h)}(M^1_{\rho, h}(x), M^2_{\rho, h}(x))
= \begin{pmatrix}
-\EE{W_i C_i \mid X_i = x} & -\EE{W_i \mid X_i = x} \\
-\EE{C_i \mid X_i = x} & -1
\end{pmatrix}  
\end{equation}
The determinant is
\[
\det V(x)
= \Cov{C_i, W_i \mid X_i = x}.
\]
Under Assumptions \ref{as:c} and \ref{assu:unconf} we have
\[
\Cov{C_i, W_i \mid X_i = x} = h_w(x)(1-h_w(x))\tau_C(x) > 0.
\]
Thus $V(x)$ is invertible for any $x$.

\paragraph{Verification of Assumption 3}

Let $\xi = (\rho, h)$ and $\xi' = (\rho', h')$ be two sets of parameters. We need to show that $\|\text{Var}(\psi_{\xi}(W_i, Y_i, C_i) - \psi_{\xi'}(W_i, Y_i, C_i) \mid X_i = x)\|_2$ is bounded by $L \|\xi - \xi'\|^2$.
\begin{equation*}
\psi_{\xi} - \psi_{\xi'}
= -(\xi - \xi') \begin{pmatrix} W_i C_i & C_i \\ W_i & 1 \end{pmatrix}
\end{equation*}

Since the random variables $W_i, C_i$ are bounded (Assumption \ref{assu:regularity}), we can bound the difference uniformly in $X_i$:
$$ \|\psi_{\xi} - \psi_{\xi'}\| \leq K \|\xi - \xi'\| $$
for some constant $K$. Then
\begin{align*}
\text{Var}(\psi_{\xi} - \psi_{\xi'} \mid X_i = x) &\leq \mathbb E[\|\psi_{\xi} - \psi_{\xi'}\|^2 \mid X_i = x] \\
&\leq K^2 \|\xi - \xi'\|^2.
\end{align*}

\paragraph{Verification of Assumption 4}

The assumption requires that the score function can be decomposed as $\psi_{\xi}(O) = \lambda(\xi; O) + \zeta_{\xi}(g(O))$, where $\lambda$ is Lipschitz continuous in the parameters $\xi = (\rho, h)$, $g(O)$ is a scalar summary, and $\zeta$ is a monotone function.

The score function $\psi_{\xi} = (\psi^1_{\xi}, \psi^2_{\xi})$ is affine in the parameters $\rho$ and $h$:
\begin{align*}
\psi^1_{\xi}(W_i, Y_i, C_i) &= W_i Y_i - \rho W_i C_i - h W_i \\
\psi^2_{\xi}(W_i, Y_i, C_i) &= Y_i - \rho C_i - h
\end{align*}
For any value of $W_i$, $C_i$, and $Y_i$, the score functions are Lipschitz continuous in the parameters. Thus we can set $\zeta=0$ and $\lambda = \psi_{\xi}$, which satisfies the assumption.

\paragraph{Verification of Assumption 5}

The assumption states that for any weights $\alpha_i$ with $\sum_{i=1}^n \alpha_i = 1$, the estimating equation returns a minimizer $(\hrho, \hat{h})$ that at least approximately solves the equations: $\|\sum_{i=1}^n \alpha_i \psi_{\hrho, \hat{h}}(O_i)\|_2 \leq C \max\{\alpha_i\}$ for some constant $C \geq 0$.

IV moment conditions always have at least one exact solution $(\hrho, \hat{h})$ for any weights $\alpha$, therefore $\sum_{i=1}^n \alpha_i \psi_{\hrho, \hat{h}}(O_i) = 0$. Thus, the condition is satisfied with $C = 0$.

\paragraph{Verification of Assumption 6}

The assumption states that 1) the functions $\psi$ are subdifferentials of some convex function, while 2) the moments $M$ are derivatives of some strongly convex function. Looking at the Jacobian \eqref{eq:jacobian_grf}, we can see that it is not symmetric, therefore this assumption must fail.

However, the result of Theorem 5 of \citet{athey2019generalized} still applies, because this assumption is only used in the proof of Lemma 10 of \citet{athey2019generalized}, which is used to prove the consistency result stated in Theorem 3 of \citet{athey2019generalized}. We can verify this directly without relying on the assumption.

To state the result from Lemma 10 of \citet{athey2019generalized}, we reiterate some of their notation:
\[
\Psi(\rho, h) := \sum_{i=1}^n \alpha_i(x) \psi_{\rho, h}(O_i) \quad \text{and} \quad \overline{\Psi}(\rho, h) := \sum_{i=1}^n \alpha_i(x) M_{\rho, h}(X_i).
\]
Substituting the definitions of $\psi$ and $M$, and writing $\xi := (\rho, h)'$, we can rewrite $\Psi(\rho, h)$ as follows:
\begin{equation}
\Psi(\xi) = b_n - A_n \xi
\end{equation}
\begin{equation}
b_n := \sum_{i=1}^n \alpha_i(x) \begin{pmatrix} W_i Y_i \\ Y_i \end{pmatrix}, \quad A_n := \sum_{i=1}^n \alpha_i(x) \begin{pmatrix} W_i C_i & W_i \\ C_i & 1 \end{pmatrix}.
\end{equation}
Similarly, $\bar \Psi(\rho, h)$ can be written as
\begin{equation}
\overline{\Psi}(\xi) = \overline{b}_n - \overline{A}_n \xi
\end{equation}
\begin{equation}
\overline{b}_n := \sum_{i=1}^n \alpha_i(x) \begin{pmatrix} \EE{W_i Y_i \mid X_i} \\ \EE{Y_i \mid X_i} \end{pmatrix}, \quad \overline{A}_n := \sum_{i=1}^n \alpha_i(x) \begin{pmatrix} \EE{W_i C_i \mid X_i} & \EE{W_i \mid X_i} \\ \EE{C_i \mid X_i} & 1 \end{pmatrix}.
\end{equation}

Now we need to show that all approximate solutions to $\Psi(\xi) = 0$ are close to each other in the following sense: for any sequence $\epsilon_n > 0$ with $\lim_{n \to \infty} \epsilon_n = 0$,
\begin{equation}
\tag{34}
\sup \left\{ \|\xi-\xi'\|_2 : \|\Psi(\xi)\|_2, \|\Psi(\xi')\|_2 < \epsilon_n \right\} \to_p 0.
\end{equation}
Here is the proof of this statement in our setting:
\[
\|\bar A_n(\xi-\xi')\|_2 \leq \|A_n(\xi-\xi')\|_2 + \|(\bar A_n - A_n)(\xi-\xi')\|_2
\]
$\|A_n(\xi-\xi')\|_2 = \|\Psi(\xi)-\Psi(\xi')\|_2$ is bounded by $2\epsilon_n$ by the premise of the Lemma. $\|(\bar A_n - A_n)(\xi-\xi')\|_2 = o_p(1)$ by Lemma 9 of \cite{athey2019generalized} under Assumptions 1–4 of \cite{athey2019generalized}, which are verified above.

Therefore, to conclude that $\|\xi-\xi'\|_2\to_p 0$, it is enough to show that the smallest singular value of $\bar A_n$ is bounded away from $0$ with probability tending to $1$.

We will show that in two steps. First, at the point $x$ the Jacobian is not singular, as verified in \eqref{eq:jacobian_grf}, which means that there is some positive constant that bounds the singular values from below. By Assumption \ref{as:lipschitz} the Jacobian is Lipschitz continuous, so there exists a neighborhood $B$ of $x$, such that on the set $B$ the singular values of the Jacobian are uniformly bounded from below with some positive constant. Secondly, by Theorem 3 of \citet{wager2018estimation} the weights are localized: $\mathbb{E}\!\left[\sup\left\{\|X_i-x\|_2:\alpha_i(x)>0\right\}\right]
= o(1)$. Therefore, in a large enough sample the weights $\alpha$ put a non-zero mass on $B$, and $\bar A_n$ singular values are bounded from below with some positive constant with probability tending to 1.

\subsection*{Proof of Theorem \ref{thm:Qexp}}

We first prove a couple of useful Lemmas.

\begin{lemma}
The estimated threshold converges to the true threshold $\hat s(b) \rightarrow_p s(b)$ and has an asymptotically linear representation:
\[ \sqrt n (\hat s(b) - s(b)) = - \frac{\sqrt n}{ G'_{\hat \rho}(s(b))} \Big (\hat G_{\hat \rho}(s(b)) - G_{\hat \rho}(s(b)) \Big) + o_p(1). \]
\label{lem:thresconvg}
\end{lemma}

\begin{proof} We can define $\hat s(b)$ as a Z-estimator, where it is the possibly non-unique and approximate solution to
\[ \hat G_{\hat \rho}(\hat s(b)) - b = 0.  \]
We can then use Theorem 5.9 of \citet{van2000asymptotic} to prove that $\hat s(b) \rightarrow_p s(b)$.  Using this Lemma requires verifying two conditions:

First, the uniform convergence of $\hat G_{\hat \rho}(s) - b$ to $G_{\hat \rho}(s) - b$ follows from Lemma 2.4 of \citet{newey1994large}. We have continuity with probability 1 in $s$ and boundedness of $G_i(s) = \left (\frac{W_i}{p} -  \frac{1 - W_i}{1- p} \right) C_i \mathbbm {1} (S_i \geq s)$ in $s$ (given that $S_i$ is continuously distributed) and that $s$ is an element of a compact space. $||\cdot||_2$ is the $L_2$ norm. 
\[ \sup_{ s \in \mathcal S}  \| \hat G_{\hat \rho}(s)  - G_{\hat \rho}(s)  \|_2  \rightarrow_p 0. \]

Next, we note that $G_{\hat \rho}(s) - b$ is continuous in $s$, $s \in \mathcal S$, which is a compact space, and $G_{\hat \rho}(s) - b$ has a unique zero at $s(b)$ since $G_{\hat \rho}(s)$ is strictly monotonic, so has an inverse. To show that $G_{\hat \rho}(s)$ is strictly monotonic, note that $G'_{\hat \rho}(s) = - f(s)E[C_i(1) - C_i(0) | S_i =s ] <0$ by the assumption that $f(s) > 0$ and Assumption \ref{as:c}. 

 This shows the second condition of Theorem 5.9 of \citet{van2000asymptotic} (see Problem 5.27):
\[ \inf_{s: d(s, s(b)) \geq \epsilon} \| G_{\hat \rho}(s) - b \|_2 > 0 =  \|  G_{\hat \rho} (s(b)) - b  \|_2.  \] 

We have now verified the conditions of Theorem 5.9 and shown that $\hat s(b) \rightarrow_p s(b)$.
\end{proof}

\begin{lemma}
The following convergence in probability holds:
\begin{enumerate}
\item $\sqrt n\left( [\hat V_{\hat \rho}(\hat s(b)) - V_{\hat \rho}(\hat s(b))] - [\hat V_{\hat \rho}( s(b)) - V_{\hat \rho}(s(b))] \right) \rightarrow_p 0$
\item $\sqrt n\left( [\hat G_{\hat \rho}(\hat s(b)) - G_{\hat \rho}(\hat s(b))] - [\hat G_{\hat \rho}( s(b)) - G_{\hat \rho}(s(b))] \right) \rightarrow_p 0$
\end{enumerate}
\label{lem:helper}
\end{lemma}

\begin{proof}
We use Lemma 19.24 of \citet{van2000asymptotic}. Given that we have shown in the previous Lemma that $\hat s(b) \rightarrow_p s(b)$, then the convergence in probability that we require holds as long as the following two conditions hold:
 \begin{enumerate}
 \item Define the function classes
 \[ \mathcal F^V = \left \{  (W_i, X_i, Y_i, S_i) \mapsto  \left ( \frac{W_i}{h_w(X_i)} - \frac{ 1- W_i}{ 1- h_w(X_i)} \right) Y_i \mathbbm {1}(S_i \geq s) : s \in \mathcal S \right \}, \]
 \[ \mathcal F^G = \left \{ (W_i, X_i, C_i, S_i) \mapsto   \left ( \frac{W_i}{h_w(X_i)} - \frac{ 1- W_i}{ 1- h_w(X_i)} \right) C_i \mathbbm {1}(S_i \geq s) : s \in \mathcal S \right \}. \]
$\mathcal F^V$ and $\mathcal F^G$ are $P$-Donsker, where $P$ defines the probability distribution of $S_i, W_i, Y_i, C_i$.
 \item $\mathbb E\left [ \Big( V_{\hat \rho}(\hat s(b) ) - V_{\hat \rho}(s(b))\Big)^2  \right ] \rightarrow_p 0 $ and $\mathbb E\left [ \Big( G_{\hat \rho}(\hat s(b) ) - G_{\hat \rho}(s(b))\Big)^2  \right ] \rightarrow_p 0$.
 \end{enumerate}

\textbf{Showing Condition 1.}

Both function classes are $P$-Donsker by overlap and the boundedness of $Y_i$ and $C_i$, and the fact that indicator functions are a Donsker class, by a bracketing argument (see, for example, Example 19.6 of \citet{van2000asymptotic}). 

\textbf{Showing Condition 2.}

Condition 2 follows from the convergence in probability of $\hat s(b)$ since $V_{\hat \rho}(s)$ and $G_{\hat \rho}(s)$ are continuous in $s$, and bounded (by the dominated convergence theorem). 

 \end{proof}

Next, for the asymptotically linear representation, we use  Theorem 5.21 of \citet{van2000asymptotic}. Lemma \ref{lem:helper} gives the required asymptotic expansion. By the continuous differentiability and strict monotonicity of $G_{\hat \rho}(s)$, we meet the required differentiability condition and that $G'_{\hat \rho}(s(b)) \neq 0$. 

Now that we have verified these conditions, then : 
\[ \sqrt n (\hat s(b) - s(b)) = \frac{- \sqrt n}{ G'_{\hat \rho}(s(b))} \Big (\hat G_{\hat \rho}(s(b)) - G_{\hat \rho}(s(b)) \Big) + o_p(1). \]

The following expansion holds for $\hat Q_{\hat \rho}(b)$ under the Assumptions of Theorem \ref{thm:Qexp} 
\begin{align}
  \hat Q_{\hat \rho}(b) - Q_{\hat \rho}(b) & = \hat V_{\hat \rho}(\hat s(b)) - V_{\hat \rho}(s(b))  \\
& =  \hat V_{\hat \rho}(\hat s(b)) - V_{\hat \rho}(\hat s(b)) + V_{\hat \rho}(\hat s(b)) - V_{\hat \rho}(s(b)) \label{eq:sq2} \\
& = \hat V_{\hat \rho}( s(b)) - V_{\hat \rho}( s(b)) + V_{\hat \rho}(\hat s(b)) - V_{\hat \rho}(s(b)) +   o_p\left(n^{-0.5} \right) \label{eq:sq3}
\end{align}
For  (\ref{eq:sq3}), we applied Lemma \ref{lem:helper}. Next, since we have that $V_{\hat \rho}(s)$ is continuously differentiable in $s$, we can use the mean-value form of a Taylor Expansion of $V_{\hat \rho}(\hat s(b))$ around $s(b)$ to show that: 

\[ V_{\hat \rho} (\hat s(b)) - V_{\hat \rho} (s(b)) = V_{\hat \rho}'(s) (\hat s(b) - s(b)) + o_p(n^{-0.5}). \]

Next, we plug in the expansion from Lemma \ref{lem:thresconvg} for an expansion for $\hat Q_{\hat \rho}(b)$: 
\[ \hat V_{\hat \rho} (\hat s(b)) - V_{\hat \rho} (s(b)) =  \hat V_{\hat \rho}( s(b)) - V_{\hat \rho}(s(b)) - \frac{V_{\hat \rho}'(s(b)) }{G_{\hat \rho}'(s(b)) } \Big( \hat G_{\hat \rho}(s(b)) - G_{\hat \rho}(s(b)) \Big )+ o_p(n^{-0.5}), \]

The RHS of the expression for $\hat Q_{\hat \rho}(b)$ is an i.i.d. average  with finite variance so the central limit theorem applies and $\hat Q_{\hat \rho}(b)$ is asymptotically normal. 

Next we convert $\hat \Delta$ to an asymptotically linear representation.
\[ \hat \Delta_{\hat \rho}(b) = \hat Q_{\hat \rho}(b) - b \frac{ \hat \tau_y}{ \hat \tau_c}. \]
Let $f(x, y) =b \frac{x}{ y}$. Take a Taylor expansion of $f(\hat \tau_y, \hat \tau_c)$ around $(\tau_y, \tau_c)$, recognizing that $(\hat \tau_y - \tau_y) = O_p(n^{-1/2})$ and $(\hat \tau_c - \tau_c) = O_p(n^{-1/2})$ by the CLT, given i.i.d. and bounded outcomes and costs: 

\[ \hat \Delta_{\hat \rho}(b) - \Delta_{\hat \rho}(b) = \hat Q_{\hat \rho}(b) - Q_{\hat \rho}(b) - b \frac{\hat \tau_y - \tau_y}{\tau_c} + b \frac{ \tau_y (\hat \tau_c - \tau_c)} {\tau_c^2} + o_p(n^{-0.5}). \]
 This now gives an expression for $\hat \Delta_{\hat \rho}(b)$ in terms of an i.i.d. average which is asymptotically normal.

\section{Additional Tables and Figures}

\label{appndx:empirical}

\begin{table}[t] \centering 
  \resizebox{\textwidth}{!}{
\begin{tabular}{|l|ccc|ccc|}
\hline
 & $\hat \Delta$ & Standard deviation & Coverage & $\hat B$ & standard deviation & \% violations \\
\hline
Instrumental Forest & $0.496$ & $0.081$ & $0.966$ & $1.000$ & $0$ & $0$ \\
Linear Parametric & $0.494$ & $0.081$ & $0.970$ & $1.000$ & $0$ & $0$ \\
Direct Ratio & $0.412$ & $0.079$ & $0.962$ & $1.000$ & $0$ & $0$ \\
Ignore Costs & $0.179$ & $0.078$ & $0.968$ & $1.000$ & $0$ & $0$ \\
\cite{hoch2002something} & $0.497$ & $$ & $$ & $0.996$ & $0.133$ & $1.000$ \\
\cite{sun2021empirical} & $0.332$ & $$ & $$ & $0.789$ & $0.100$ & $0.996$ \\ 
\hline
\end{tabular}
}
\vspace{0.2\baselineskip}

  \caption{The table shows the performance of different methods in the partially predictable costs simulation, under a budget constraint of 1, trained on a sample of 500 observations. $\hat \Delta$ is the estimated lift of the reward over the uniform allocation  from a sample of 1000 individuals, averaged over 500 simulation replicates. The next column shows the half-sample bootstrapped (1,000 bootstrap samples) standard deviations of $\hat \Delta$ averaged across 500 simulation replications and the coverage of the $(\hat \Delta - 1.96 \operatorname{se}(\hat \Delta), \hat \Delta + 1.96 \operatorname{se}(\hat \Delta))$ confidence interval, where the ground truth was computed via simulation.  Standard errors for direct optimization methods are not currently available in the literature. We also report the average budget spent and its standard deviation, as well as the percentage of simulation replicates for which the budget spent in the deployment set is higher than 1.  } 
  \label{tab:simu} 
\end{table}

\begin{table}[!ht]
    \centering
    \resizebox{\textwidth}{!}{
    \begin{tabular}{|l|l|l|}
    \hline
    Variable name & Variable description & Included in parametric \\
    \hline \hline
    numhh\_list & Number of people in household on lottery list & No \\ \hline
    birthyear\_list & Birth year: lottery list data & Yes \\ \hline
    have\_phone\_list & Gave a phone number on lottery sign up: lottery list data & No \\ \hline
    english\_list & Individual requested english-language materials: lottery list data & No \\ \hline
    female\_list & Female: lottery list data & No \\ \hline
    first\_day\_list & Signed up for lottery list on first day: lottery list data & No \\ \hline
    last\_day\_list & Signed up for lottery list on last day: lottery list data & No \\ \hline
    pobox\_list & Gave a PO Box as an address: lottery list data & No \\ \hline
    self\_list & Individual signed him or herself up for the lottery list & No \\ \hline
    zip\_msa\_list & Zip code from lottery list is a metropolitan statistical area & No \\ \hline
    snap\_ever\_presurvey12m & Ever personally on SNAP, 6 month pretreatment & No \\ \hline
    snap\_tot\_hh\_presurvey12m & Total household benefits from SNAP, 6 month pretreatment & Yes \\ \hline
    tanf\_ever\_presurvey12m & Ever personally on TANF, 6 month pretreatment & No \\ \hline
    tanf\_tot\_hh\_presurvey12m & Total household benefits from TANF, 6 month pretreatment & No \\ \hline
    any\_visit\_pre\_ed & Any ED visit,  & No \\ \hline
    any\_hosp\_pre\_ed & Any ED visit resulting in a hospitalization  & No \\ \hline
    any\_out\_pre\_ed & Any Outpatient ED visit  & No \\ \hline
    any\_on\_pre\_ed & Any weekday daytime ED visit  & No \\ \hline
    any\_off\_pre\_ed & Any weekend or nighttime ED visits & No \\ \hline
    num\_edcnnp\_pre\_ed & Number of emergent, non-preventable ED visits  & No \\ \hline
    num\_edcnpa\_pre\_ed & Number of emergent, preventable ED visits  & No \\ \hline
    num\_epct\_pre\_ed & Number of primary care treatable ED visits  & No \\ \hline
    num\_ne\_pre\_ed & Number of non-emergent ED visits  & No \\ \hline
    num\_unclas\_pre\_ed & Number of of unclassified ED visits  & No \\ \hline
    any\_acsc\_pre\_ed & Any ambulatory case sensitive ED visit & No \\ \hline
    any\_chron\_pre\_ed & Any ED visit for chronic condition & No \\ \hline
    any\_inj\_pre\_ed & Any ED visit for injury  & No \\ \hline
    any\_skin\_pre\_ed & Any ED visit for skin conditions  & No \\ \hline
    any\_abdo\_pre\_ed & Any ED visit for abdominal pain  & No \\ \hline
    any\_back\_pre\_ed & Any ED visit for back pain  & No \\ \hline
    any\_heart\_pre\_ed & Any ED visit for chest pain  & No \\ \hline
    any\_head\_pre\_ed & Any ED visit for headache  & No \\ \hline
    any\_depres\_pre\_ed & Any ED visit for mood disorders  & No \\ \hline
    any\_psysub\_pre\_ed & Any ED visit for psych conditions/substance abuse  & No \\ \hline
    charg\_tot\_pre\_ed & Sum of total charges & Yes \\ \hline
    ed\_charg\_tot\_pre\_ed & Sum of total ED charges  & Yes \\ \hline
    any\_hiun\_pre\_ed & Any ED visit to a high uninsured volume hospital  & No \\ \hline
    any\_loun\_pre\_ed & Any ED visit to a low uninsured volume hospital  & No \\ \hline
    need\_med\_0m & Survey data: Needed medical care in the last six months & No \\ \hline
    need\_rx\_0m & Survey data: Needed prescription medications in the last six months & No \\ \hline
    rx\_num\_mod\_0m & Survey data: Number of prescription medications currently taking & Yes \\ \hline
    rx\_any\_0m & Survey data: Currently taking any prescription medications & Yes \\ \hline
    need\_dent\_0m & Survey data: Needed dental care in the last six months & Yes \\ \hline
    doc\_any\_0m & Survey data: Any primary care visits & No \\ \hline
    doc\_num\_mod\_0m & Survey data: Number of primary care visits, truncated & Yes \\ \hline
    er\_any\_0m & Survey data: Any ER visits & No \\ \hline
    er\_num\_mod\_0m & Survey data: Number of ER visits, truncated & No \\ \hline
    er\_noner\_0m & Survey data: Used emergency room for non-emergency care & No \\ \hline
    reason\_er\_need\_0m & Survey data: Went to ER (reason): needed emergency care & No \\ \hline
    \end{tabular}}
\end{table}

\begin{table}[!ht]
    \centering
    \resizebox{\textwidth}{!}{
    \begin{tabular}{|l|l|l|}
    \hline
    Variable name & Variable description & Included in parametric \\
    \hline \hline
    reason\_er\_closed\_0m & Survey data: Went to ER (reason): clinics closed & No \\ \hline
    reason\_er\_apt\_0m & Survey data: Went to ER (reason): couldn't get doctor's appointment & No \\ \hline
    reason\_er\_doc\_0m & Survey data: Went to ER (reason): didn't have personal doctor & No \\ \hline
    reason\_er\_copay\_0m & Survey data: Went to ER (reason): couldn't afford copay to see a doctor  & No \\ \hline
    reason\_er\_go\_0m & Survey data: Went to ER (reason): didn't know where else to go & No \\ \hline
    reason\_er\_other\_0m & Survey data: Went to ER (reason): other reason & No \\ \hline
    reason\_er\_rx\_0m & Survey data: Went to ER (reason): needed prescription drug & No \\ \hline
    reason\_er\_dont\_0m & Survey data: Went to ER (reason): don't know & No \\ \hline
    hosp\_any\_0m & Survey data: Any hospital visits & No \\ \hline
    hosp\_num\_mod\_0m & Survey data: Number hospital visits, truncated at 2*99th\%ile & No \\ \hline
    total\_hosp\_0m & Survey data: Total days spent in hospital, last 6 months & Yes \\ \hline
    dia\_dx\_0m & Survey data: Diagnosed diabetes & No \\ \hline
    ast\_dx\_0m & Survey data: Diagnosed asthma & No \\ \hline
    hbp\_dx\_0m & Survey data: Diagnosed high blood pressure & No \\ \hline
    emp\_dx\_0m & Survey data: Diagnosed COPD & No \\ \hline
    chf\_dx\_0m & Survey data: Diagnosed congestive heart failure & No \\ \hline
    dep\_dx\_0m & Survey data: Diagnosed depression or anxiety & Yes \\ \hline
    female\_0m & Survey data: Is female & No \\ \hline
    birthyear\_0m & Survey data: Birth year & No \\ \hline
    employ\_0m & Survey data: Currently employed & No \\ \hline
    employ\_det\_0m & Survey data: Currently employed or self-employed & Yes \\ \hline
    hhinc\_cat\_0m & Survey data: Household income category & Yes \\ \hline
    employ\_hrs\_0m & Survey data: Average hrs worked/week & Yes \\ \hline
    edu\_0m & Survey data: Highest level of education completed & Yes \\ \hline
    living\_arrange\_0m & Survey data: Current living arrangement & Yes \\ \hline
    hhsize\_0m & Survey data: Household Size (adults and children) & Yes \\ \hline
    hhinc\_pctfpl\_0m & Survey data: Household income as percent of federal poverty line & Yes \\ \hline
    num19\_0m & Survey data: Number of family members under 19 living in house & Yes \\ \hline
    preperiod\_any\_visits & Any ED visit (the date range is different from any\_visit\_pre\_ed) & No \\ \hline
    \end{tabular}}
    \caption{List of variables used as pre-treatment covariates in the Oregon Health Experiment application}
\label{table:vars}
\end{table}
\newpage

\label{appndx:n}
\begin{figure}
\centering
\begin{tabular}{cc}
\includegraphics[width=0.4\textwidth]{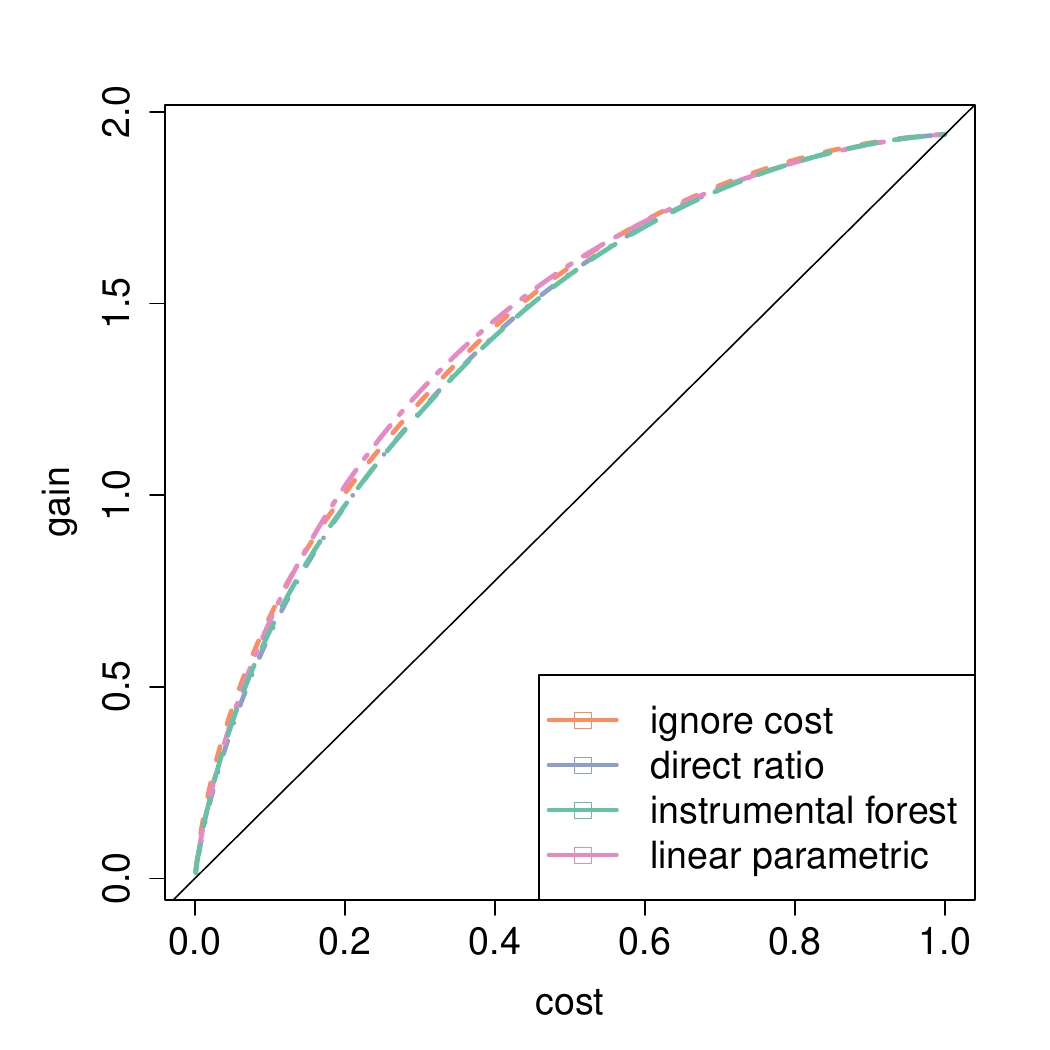} &
\includegraphics[width=0.4\textwidth]{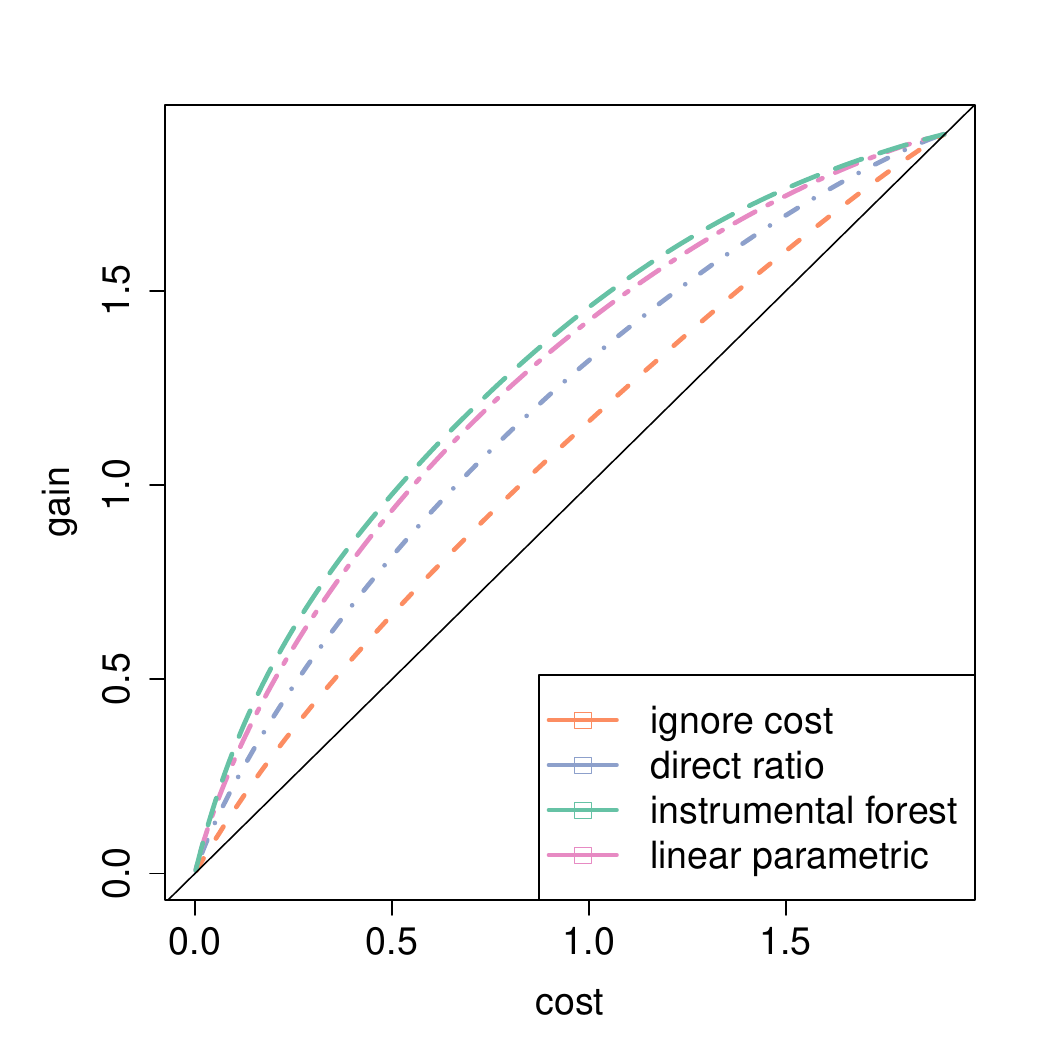} \\
unpredictable cost, $n=200$ & partially predictable cost, $n=200$ \\
\includegraphics[width=0.4\textwidth]{figures/simu_nocost_500.pdf} &
\includegraphics[width=0.4\textwidth]{figures/simu_main_500.pdf} \\
unpredictable cost, $n=500$ & partially predictable cost, $n=500$ \\
\includegraphics[width=0.4\textwidth]{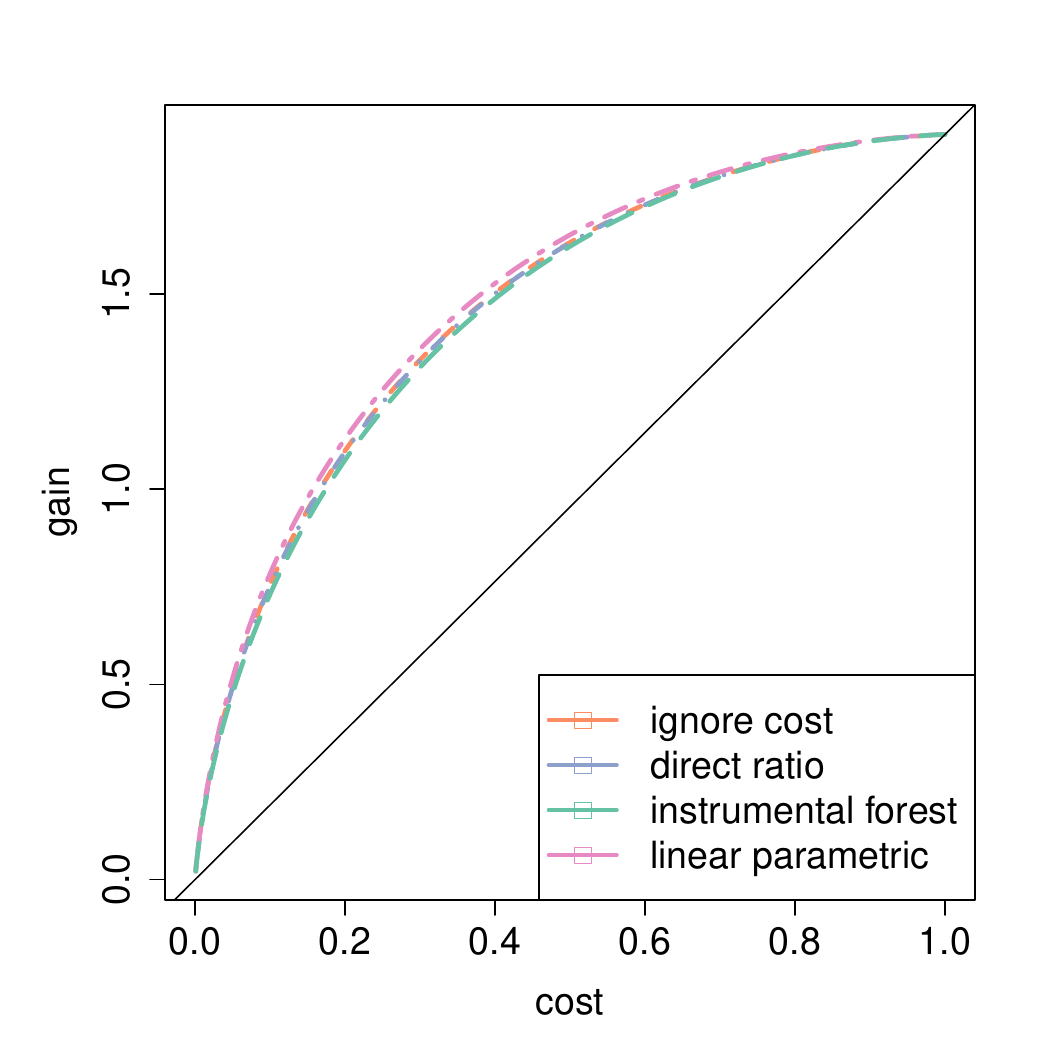} &
\includegraphics[width=0.4\textwidth]{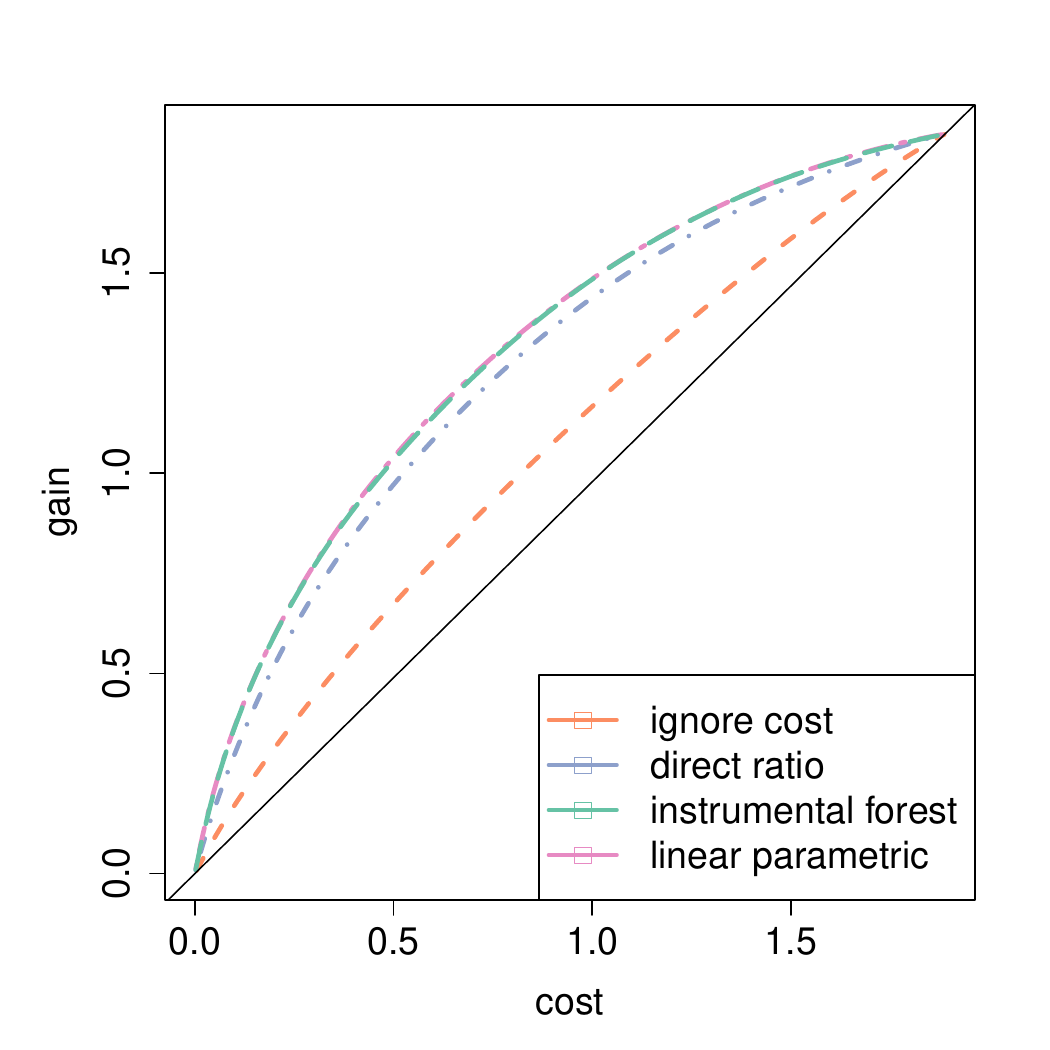} \\
unpredictable cost, $n=1000$ & partially predictable cost, $n=1000$
\end{tabular}
\caption{QINI curves for the simulation settings with predictable and unpredictable costs, averaged over
500 simulation replicates. Comparing to Figure \ref{fig:simu} we vary the training samples between $n \in \{200, 500, 1000\}$.}
\label{fig:simu_appndx}
\end{figure}

\newpage
\begin{table}[h!] \centering 
  \resizebox{\textwidth}{!}{
\begin{tabular}{|l|cccc|cc|c|} 
\hline 
 & $\hat \Delta$ & Standard deviation & Coverage & $\hat B$ & standard deviation & \% violations & training sample\\
\hline
Instrumental forest & $0.507$ & $0.081$ & $0.980$ & $1.000$ & $0$ & $0$ & $1,000$ \\
Linear parametric & $0.510$ & $0.081$ & $0.974$ & $1.000$ & $0$ & $0$ & $1,000$ \\
Direct Ratio & $0.457$ & $0.080$ & $0.976$ & $1.000$ & $0$ & $0$ & $1,000$ \\
Ignore Costs & $0.197$ & $0.077$ & $0.978$ & $1.000$ & $0$ & $0$ & $1,000$ \\
\cite{hoch2002something} & $0.509$ & $$ & $$ & $1.000$ & $0.134$ & $1.000$ & $1,000$ \\
\cite{sun2021empirical} & $0.344$ & $$ & $$ & $0.796$ & $0.104$ & $1.000$ & $1,000$ \\ 
\hline
Instrumental forest & $0.496$ & $0.081$ & $0.966$ & $1.000$ & $0$ & $0$ & $500$ \\
Linear parametric & $0.494$ & $0.081$ & $0.970$ & $1.000$ & $0$ & $0$ & $500$ \\
Direct Ratio & $0.412$ & $0.079$ & $0.962$ & $1.000$ & $0$ & $0$ & $500$ \\
Ignore Costs & $0.179$ & $0.078$ & $0.968$ & $1.000$ & $0$ & $0$ & $500$ \\
\cite{hoch2002something} & $0.497$ & $$ & $$ & $0.996$ & $0.133$ & $1.000$ & $500$ \\
\cite{sun2021empirical} & $0.332$ & $$ & $$ & $0.789$ & $0.100$ & $0.996$ & $500$ \\ 
\hline
Instrumental forest & $0.460$ & $0.080$ & $0.970$ & $1.000$ & $0$ & $0$ & $200$ \\
Linear parametric & $0.432$ & $0.080$ & $0.978$ & $1.000$ & $0$ & $0$ & $200$ \\
Direct Ratio & $0.321$ & $0.079$ & $0.970$ & $1.000$ & $0$ & $0$ & $200$ \\
Ignore Costs & $0.153$ & $0.079$ & $0.964$ & $1.000$ & $0$ & $0$ & $200$ \\
\cite{hoch2002something} & $0.475$ & $$ & $$ & $1.010$ & $0.147$ & $1.000$ & $200$ \\
\cite{sun2021empirical} & $0.308$ & $$ & $$ & $0.804$ & $0.115$ & $0.998$ & $200$ \\ 
\hline
\end{tabular}
}
  
  \caption{The table shows the performance of different methods in the partially predictable costs simulation, under a budget constraint of 1 and for training sample $n \in \{200, 500, 1000\}$. $\hat \Delta$ is the estimated lift of the reward over the uniform allocation from a sample of 1000 individuals, averaged over 500 simulation replicates. The next column shows the half-sample bootstrapped (1,000 bootstrap samples) standard deviation of $\hat \Delta$ averaged across 500 simulation replications and the coverage of the $(\hat \Delta - 1.96 \operatorname{se}(\hat \Delta), \hat \Delta + 1.96 \operatorname{se}(\hat \Delta))$ confidence interval, where the ground truth was computed via simulation.  Standard errors for direct optimization methods are not currently available in the literature. We also report the average budget spent and its standard deviation, as well as the percentage of simulation replicates for which the budget spent in the deployment set is higher than 1.} 
  \label{tab:simu_apndx} 

\end{table}

\newpage
\begin{table} \centering 
\begin{tabular}{lcc}
\hline
& \multicolumn{2}{c}{\textit{Cost variable:}} \\
\cline{2-3}
& Medications & Outpatient visits \\
& (1) & (2) \\
\hline
Instrumental forest & 0.0178 & 0.0083 \\
Direct ratio & 0.0131 & 0.0070 \\
Ignore costs & 0.0071 & 0.0063 \\
Linear parametric & 0.0145 & 0.0025 \\
\hline
\end{tabular}
\caption{Area under the curve metric. It is calulated as the area between the uniform allocation line and the QINI curve of the respective metric.}
\label{table:roc}
\end{table}

\end{document}